\documentclass[12pt]{article}
\usepackage{amsmath}
\usepackage{graphicx}
\usepackage{enumerate}
\usepackage{natbib}
\usepackage{url} 
\usepackage[OT1]{fontenc} 
\usepackage{times}
\usepackage{float}
\usepackage{amsthm}
\usepackage{amsfonts}
\usepackage{algorithm}
\usepackage{makecell}
\usepackage{multirow}
\usepackage{hhline}
\usepackage{algorithmic}
\usepackage[shortlabels]{enumitem}
\usepackage{mathrsfs}

\def\bs{\boldsymbol}

\newcommand{\Ac}{\mathcal{A}}

\newcommand{\B}{\mathbf{B}}
\newcommand{\Bt}{\widetilde{\mathbf{B}}}

\newcommand{\Cc}{\mathcal{C}}

\newcommand{\D}{\mathbf{D}}

\newcommand{\Ec}{\mathcal{E}}

\newcommand{\Gc}{\mathcal{G}}
\newcommand{\g}{\mathbf{g}}
\newcommand{\gp}{g^\prime}

\newcommand{\K}{\mathbf{K}}

\newcommand{\Mcc}{\mathcal{M}}

\newcommand{\Pc}{\mathcal{P}}

\renewcommand{\P}{\mathbf{P}}
\newcommand{\Q}{\mathbf{Q}}

\newcommand{\R}{\mathbf{R}}
\newcommand{\Rc}{\mathbb{R}}

\newcommand{\s}{\mathbf{s}}

\newcommand{\Tc}{\mathcal{T}}

\newcommand{\Vc}{\mathcal{V}}

\newcommand{\w}{\mathbf{w}}

\newcommand{\X}{\mathbf{X}}

\newcommand{\x}{\mathbf{x}}

\newcommand{\y}{\mathbf{y}}

\newcommand{\0}{\mathbf{0}}
\newcommand{\1}{\mathbf{1}}

\renewcommand{\H}{\mathbf{H}}

\newcommand{\tr}{\mathrm{Tr}}

\newtheorem{thm}{Theorem}

\newtheorem{lem}{Lemma}

\newcommand{\Valpha}{\boldsymbol \alpha}
\newcommand{\vbeta}{\boldsymbol \beta}

\newcommand{\vpsi}{\boldsymbol \psi}

\newcommand{\vtheta}{\boldsymbol \theta}

\newcommand{\Diag}{\mathop{\mathrm{diag}}}

\newcommand{\blind}{1}

\begin{document}

\def\spacingset#1{\renewcommand{\baselinestretch}%
{#1}\small\normalsize} \spacingset{1}


\if1\blind
{
  \title{\bf Spatial Heterogeneous Additive Partial Linear Model:\\ A Joint Approach of Bivariate Spline and Forest Lasso}
  \author{Xin Zhang\\
    Department of Statistics, Iowa State University\\
    Shan Yu \\
   Department of Statistics, University of Virginia\\
     Zhengyuan Zhu\\
     Department of Statistics, Iowa State University\\
     Xin Wang\\
    Department of Mathematics and Statistics, San Diego State University\\
     }
     \date{}
  \maketitle
} \fi

\if0\blind
{
  \bigskip
  \bigskip
  \bigskip
  \begin{center}
    {\LARGE\bf Spatial Heterogeneous Additive Partial Linear Model: A Joint Approach of Bivariate Spline and Forest Lasso}
\end{center}
  \medskip
} \fi

\bigskip
\begin{abstract}
Identifying spatial heterogeneous patterns has attracted a surge of research interest in recent years, due to its important applications in various scientific and engineering fields. In practice the spatially heterogeneous components are often mixed with components which are spatially smooth, making the task of identifying the heterogeneous regions more challenging. 
In this paper, we develop an efficient clustering approach to identify the model heterogeneity of the spatial additive partial linear model. 
Specifically, we aim to detect the spatially contiguous clusters based on the regression coefficients while introducing a spatially varying intercept to deal with the smooth spatial effect.
On the one hand, to approximate the spatial varying intercept, we use the method of bivariate spline over triangulation, which can effectively handle the data from a complex domain.
On the other hand, a novel fusion penalty termed the forest lasso is proposed to reveal the spatial clustering pattern.
Our proposed fusion penalty has advantages in both the estimation and computation efficiencies when dealing with large spatial data.
Theoretically properties of our estimator are established, and simulation results show that our approach can achieve more accurate estimation with a limited computation cost compared with the existing approaches. To illustrate its practical use, we apply our approach to analyze the spatial pattern of the relationship between land surface temperature measured by satellites and air temperature measured by ground stations in the United States.
\end{abstract}

\noindent%
{\it Keywords:}   Bivariate splines; Forest lasso; Penalized least squares; Spatial coefficient clustering; Spatial regression model.
\vfill

\newpage
\spacingset{1.5} 


\section{Introduction}\label{Section: introduction}

Large scale spatial data analysis has received significant research interests in recent years, thanks to its ever-increasing research applications in the diverse disciplines, including the fields of agriculture \citep{cho2010geographically, imran2013modeling, wang2016loss, zhang2023simulation}, geology \citep{chen2011mineral, jiang2017comparative, lewis2020modeling,zhou2020satellite}, biology \citep{dark2004biogeography, perez2010spatial, wang2020simultaneous} and social science \citep{takagi2012neighborhood, nicholson2019spatial,you2020social}, etc.
Usually, data from a large spatial domain often exhibit heterogeneity in space in the sense that the model describing the relationship among multiple spatially indexed variables may have different parameters in different regions. 
To capture such spatial variability, one widely adopted approach is to assume the spatial variation is sufficiently smooth, and fit the data with some spatially varying coefficent model \citep{fotheringham2003geographically,gelfand2003spatial,wang2016efficient,yu2019estimation,yu2022spatiotemporal}.
However, in many research applications \citep{li2019spatial,hu2020bayesian,zhu2014spatially} the spatial pattern stays homogeneous in the interior of some regions while changes abruptly at their boundaries.
In such applications, smoothing obscures the boundary, and it is often crucial to identify heterogeneous cluster patterns over the spatial domain.

In general it is a non-trivial task to detect the underlying spatial heterogeneous clusters from observations. 
First of all,  many of the existing spatial clustering approaches identify the cluster patterns based on physical spatial structures \citep{tung2001spatial} or spatial direct observations \citep{luo2003spatial}.
In many research applications, it is of great need to cluster the spatial locations based on the latent relationship between the response and explanatory variables.
For instance, due to the different land cover types, the relationship among atmosphere conditions (e.g. air temperature and humidity) and geological properties (e.g. elevation) often exhibit complex spatial cluster patterns \citep{zhang2019distributed}.
Secondly, it is necessary and reasonable to require the spatial contiguity of the identified clusters, which can facilitate further interpretations and studies.
But such spatial contiguity cannot be guaranteed without extra procedures \citep{lavigne2012model,zhang2014spatio}, or leads to complex model design and tremendous computation \citep{li2015bayesian,wang2019spatial}.
Thirdly, these underlying spatial clusters usually have highly-irregular boundaries, which cannot be captured by some of the existing methods. For example, \cite{lee2017cluster} detected the spatial clusters by testing the spatial scan statistics of pre-determined circular windows. 
But these regular-shape windows would often cause serious misclassification when irregular and complex boundaries exist.

To overcome the above limitations, a few recent works adopted the fusion penalization technique.
The key idea of this approach is to penalize the differences of adjacent model parameters, where the adjacency is defined by the spatial relationship of the observations.
By shrinking the difference to zero, the adjacent locations with the same model parameters will form the spatially contiguous clusters, which could have highly-irregular shapes.
However, designing a clustering approach with the fusion penalty commonly faces the dilemma of `speed' versus `accuracy', especially involving massive spatial locations (see Section \ref{Section: RelatedWorks} for a detailed discussion).
Meanwhile, most of the existing works only focus on the simple parametric model, i.e. the linear regression model, and identify the spatial cluster patterns based on both regression coefficients and intercept.
These models assume that the spatial effect, represented by intercept, also exhibits the cluster pattern.
However, this assumption may cause a significant loss in the estimation efficiency, especially in analyzing data with strong spatial correlation. 
Our goal of this work is to address the above challenges and develop a simple but efficient approach for spatially cluster detection.
The main results of this paper are summarized as follows:
\vspace{-.05in}
\begin{itemize}
	\item We propose a spatial heterogeneous additive partial linear model (SHAPLM) for the spatial contiguous cluster detection problem.
	Specifically, our model identifies the spatial clusters according to the coefficient of the linear components, while a common nonlinear intercept is added to deal with the spatial smoothing effect. 
	Compared with the existing parametric regression models, the proposed partial linear model has better fitting performance, especially for data with strong spatial correlation.

	\vspace{-.05in}

	\item We propose a novel fusion penalized method, termed as the forest lasso, to identify spatial heterogeneous cluster pattern of linear coefficients.
	The forest lasso is an adaptive fusion penalty, of which the adaptive weights are from the average of the estimations over multiple spanning trees.
	We show that the forest lasso can notably improve the estimation accuracy while maintaining a low computation complexity. 

	\vspace{-.05in}

	\item To estimate the nonlinear intercept, we adopt the method of bivariate spline over triangulation \citep{lai2013bivariate}. 
	This spline method can effectively handle the data from the complex domain.
	With the spline approximation, the semiparametric estimation problem reduces to a doubly penalized least square problem, which can be easily solved in numerical.
	Theoretically, we show that our proposed estimators are asymptotically consistent.

	\vspace{-.05in}

	\item In the end, our model is applied to analyze the spatial cluster pattern of the relationship between two types of temperatures in the United States.
	We reveal spatial cluster patterns based on the estimated linear coefficients at different time periods over the year.
	Compared with the existing approaches, it is shown that the estimation results of our approach have a better interpretation.
	\vspace{-.05in}
\end{itemize}
The article is organized as follows.
In Section \ref{Section: RelatedWorks}, we will review related works on the fusion penalized clustering approaches. 
In Section \ref{Section: Model}, we first formulate our spatial heterogeneous additive partial linear model (SHAPLM) and then propose our estimation procedure.
The theoretical properties of our model estimator are studied in Section \ref{Section: analysis}.
The simulation experiments and real data analysis are provided in Section \ref{Section: Simulation} and Section \ref{Section: RealData}, respectively.
In the end, Section \ref{Section: conclusion} concludes this paper.


\section{Related Works}\label{Section: RelatedWorks}

Our work mainly focuses on identifying spatial heterogeneous clusters with the fusion penalty.
In the literature, the fusion penalization approach has been broadly adopted to perform the model-based clustering analysis.
In \cite{pan2013cluster}, the penalized regression-based clustering (PRclust) approach was proposed. 
The PRclust can identify the clusters via a pairwise coefficient fusion penalization. However, this method only aims to solve the classic clustering problem of individual responses.
\cite{tang2016fused} proposed the Fused Lasso Approach in Regression Coefficients Clustering (FLARCC) to identify heterogeneity clustering patterns of study-specific effects across datasets.
FLARCC detects the clustering structure by penalizing the weighted $\ell_1$-norm differences of adjacent regression coefficients.
But to determine the adjacency and adaptive weights, the local replications are necessitated in FLARCC.
\cite{ma2017concave,ma2018explore} are the first few works achieving the latent model-based clusters identification without repeated observations.
More specifically, they focused on the linear regression model and added the concave fusion penalties, i.e., the SCAD \citep{fan2001variable} and MCP penalty \citep{zhang2010nearly}, to cluster the intercepts and the regression coefficients, respectively.
Thanks to the elegant properties of the concave penalties, the consistencies of model estimation and clustering are established in their works.
The success of the concave fusion penalties facilitates the model-based clustering analysis in other complex models.
In \cite{yang2019high} and \cite{yu2023fusion}, the authors studied the coefficient clustering problem in the high dimensional regression setting, by adopting the concave pairwise fusion penalty to identifying latent grouping structure, as well as a common sparsity penalty to perform variable selection. 
The authors of \cite{liu2019subgroup} considered a heterogeneous additive partially linear
model, which contains the homogeneous linear components and cluster-dependent nonlinear
components.
In their method, the B-spline method is adopted to approximate the nonlinear components, and the concave pairwise fusion penalty is applied to automatically identify the latent clusters of the spline coefficients.
In \cite{zhu2018cluster}, the authors extended the clustering problem in the non-parametric setting.
Their proposed method utilized the concave pairwise fusion penalty on the B-spline coefficients to partition the profiles of longitudinal data. \cite{wang2023clustering} considered the clustering problem in the framework of functional principal component analysis, also clustered units based on B-spline coefficients.
Note that the above-mentioned approaches treat the observations spatially independent and cannot handle data with the auxiliary network structure or underlying spatial restrictions.

In the field of spatial data analysis, the fusion penalized clustering approach also attracts a significant amount of recent research.
In \cite{wang2019spatial}, the authors extended the above concave fusion penalized regression approach to a class of spatially-weighted pairwise fusion method, referred to SaSa. 
The SaSa approach penalizes on the $\ell_2$-norm differences of the regression coefficients at each two locations, and incorporate the spatial information by adding a set of spatial weights.
However, because of the pairwise fusion penalties, the SaSa approach usually suffers the tremendous computation: to cluster $n$ spatial locations, the number of pairwise fusion penalties are $O(n^2)$.
To simplify the penalty terms, the authors of \cite{li2019spatial} proposed a tree-based fusion penalty to capture the spatial homogeneity between the coefficients.
Their penalties are defined by a Euclidean-distance-based minimum spanning tree (MST) and only penalize the parameter differences of two connected locations in the tree.
Nevertheless, this pre-determined MST usually depends on the spatial distance information and tends to over-cluster the spatial locations.
To address the over-clustering problem, \cite{zhang2019distributed} proposed an adaptive tree-based fusion penalty.
They determined an adaptive tree structure by combining the spatial information with the local model similarities. \cite{wang2023clustered} used the adaptive tree approach in a Poisson process model. 
However, to construct the adaptive tree, it requires the local repeated observations at each location, which is rare in spatial data analysis.
In this work, we propose a novel adaptive fusion penalty, named the forest Lasso. Our penalty is based on the repeated `tree' structures instead of the repeated observations.
It will be shown that our method has an advantage in the trade-off between the estimation accuracy and the computation speed.
Additionally, Most of the existing works focus on the parametric linear regression model.
Our work further extends the spatial coefficient clustering problem to a partial linear model, in which a spatially smoothing intercept is introduced to capture the spatial smoothing effect.


\section{Model and Method}\label{Section: Model}


In this section, we first propose our spatial heterogeneous additive partial linear model (SHAPLM), and then provide the estimation procedure.

Consider the spatial dataset $\{y(\s_i), \x(\s_i),\s_i\}_{i=1}^{n},$ where $\s_i \in \Rc^2$ presents the $i$th location, $\x(\s_i)\in \Rc^p$ and $y(\s_i)$ are the corresponding covariates and response variable at the location $\s_i.$ 
Suppose the $n$ locations $\{\s_i\}_{i=1}^{n}$ are from $K$ underlying spatially contiguous clusters $\{\Cc_k\}_{k=1}^K$.
In the $k$th cluster, the data are from the following partial linear model: 
\begin{align}\label{Eq: Model}
y(\s) = g(\s) + \x(\s)^\top\vbeta^{(k)} + \epsilon(\s), \text{if~} \s \in \Cc_k,
\end{align}
where $g(\cdot)\in\mathbb{G}$ is a unknown but smooth function based on the location and $\mathbb{G}$ is the functional space, $\vbeta^{(k)}\in\Rc^p$ is a $p$-dimensional linear coefficient vector, $\epsilon(\s)$ is the measurement error.
In the model (\ref{Eq: Model}), the response variable $y$ depends on the covariates $\x$ through a linear function that varies across the clusters and depends on the location $\s$ through a nonlinear function $g(\cdot)$ that is common to all the clusters.
Here the function $g(\cdot)$ is a spatial varying intercept to absorb the spatial random effect.
In practice, the partition $\{\Cc_k\}_{k=1}^K$ is unknown. 
Thus, the goal is to perform the model estimation and clustering simultaneously.

To achieve the goal, we consider to minimize the following penalized objective function:
\begin{align}\label{Eq: min_original}
\min_{g\in\mathbb{G},\vbeta \in \Rc^{np}} \frac{1}{2n}\sum_{i=1}^{n}\big(y(\s_i) - g(\s_i) - \x(\s_i)^\top\vbeta(\s_i)\big)^2 + \rho \mathcal{J}(g) + \lambda \mathcal{P}(\vbeta),
\end{align}
where $\vbeta = [\vbeta(\s_1)^{\top},\cdots,\vbeta(\s_n)^{\top}]^{\top}$, $\rho$ and $\lambda$ are two tuning parameters. 
The problem (\ref{Eq: min_original}) is a doubly penalized least square estimation problem. 
The first penalty $\mathcal{J}(\cdot)$ is to control the trade-off between fidelity to the data and roughness of the function.
The second function $\mathcal{P}(\cdot)$ is a penalty aiming to pursuit the homogeneity of two coefficient vectors if their corresponding locations are close.
There are two main challenges in solving the above minimization problem.
On the one hand, the nonlinear function $g(\cdot)$ is a bivariate function depending on the location $\s$.
In most cases, the observations are located in the spatial domains with complex boundaries and holes.
These complex spatial structures make the tensor product spline method and kernel smoothing method less desirable \citep{lai2013bivariate}.
Thus, it requires an efficient approach for estimating the bivariate nonlinear function $g(\cdot)$.
On the other hand, the design of the penalty function $\mathcal{P}(\cdot)$ is crucial, in terms of both the computation efficiency and estimation accuracy.
As discussed in Section \ref{Section: RelatedWorks}, the spatial pairwise penalty function \citep{wang2019spatial} could achieve accurate results but suffers the $O(n^2)$ computational cost, while the naive minimum spanning tree-based penalty function \citep{li2019spatial} enjoys the fast computation but the estimation results are sensitive to the observed locations.
This motivates us to develop a new fusion approach, which can find a balance between computation and accuracy.
In the following, we provide our approaches to address these two challenges.

\subsection{Bivariate Spline Approximation Over Triangulation}

We first approximate the spatially smooth intercept function $g(\cdot)$ by the method of bivariate penalized spline over triangulation (BPST) \citep{lai2013bivariate}.
A triangulation of a spatial domain $\Omega$ is a set of $M$ triangles $\Delta = \{\tau_1,\cdots,\tau_M\}$ if $\Omega = \cup_{m=1}^{M}\tau_m$ and any two triangles share a vertex or an edge at most.
Given a triangle $\tau \in \Delta$ and an arbitrary point $\s \in \Rc^2,$ let $b_1,$ $b_2$ and $b_3$ be the barycentric coordinates of $\s$ relative to $\tau.$
The Bernstein basis polynomials of degree $d \ge 1$ relative to triangle $\tau$ is defined as $B_{ijk}^{\tau,d}(\s) = (i!j!k!)^{-1}d!b_1^ib_2^jb_3^k,$ with $i+j+k=d.$
Denote $\mathbb{P}_d(\tau)$ as the space of all polynomials of degree less than or equal to $d$ on $\tau.$
For any integer $r\ge 0,$ let $\mathbb{C}^r(\Omega)$ be the collection of all $r$th continuously differentiable functions on over $\Omega.$
Denote $\mathbb{S}_d^r(\Delta) = \{\zeta\in\mathbb{C}^r(\Omega): \zeta|\tau \in \mathbb{P}_d(\tau),\tau\in\Delta\}$ as the spline space of degrees $d$ and smoothness {$r$} over triangulation $\Delta,$ where $\zeta|\tau$ is the polynomial piece of spline $\zeta$ restricted on triangle $\tau.$
Let $\{B_m\}_{m\in\Mcc}$ be the set of bivariate Bernstein basis polynomials for ${\mathbb{S}_{d}^{r}(\Delta)}$, where $\Mcc$ is an index set of Bernstein basis polynomials.
Then any function $\zeta \in \mathbb{S}_d^r(\tau)$ can be represented with the following basis expansion:
$\zeta(\s) = \sum_{m\in\Mcc} B_m(\s)\alpha_m = \B(\s)^\top\Valpha,$
where $\Valpha = [\alpha_m, m\in \Mcc]^\top$ is the spline coefficient vector, and $\B(\s) = [B_m(\s),m\in \Mcc]^\top$ is the evaluation vector of Bernstein basis polynomials at location $\s.$ 
The above bivariate spline basis can be easily generated via the R package \texttt{BPST}.

Adopting the BPST method, we can approximate the nonlinear function $g(\cdot)$ with a function $\gp(\cdot)\in\mathbb{S}_{d}^{r}(\Delta).$
Thus, we consider the following objective function:
\begin{align}\label{Eq: min_spline}
\min_{\gp\in\mathbb{S}_{d}^{r}(\Delta),\vbeta \in \Rc^{np}} 
\frac{1}{2n}\sum_{i=1}^{n}\big(y(\s_i) - \gp(\s_i) - \x(\s_i)^\top\vbeta(\s_i)\big)^2 + \rho \mathcal{J}(\gp) + \lambda \mathcal{P}(\vbeta),
\end{align}
where $\mathcal{J}(\gp) = \sum_{\tau\in\Delta} \int_\tau \sum_{i+j=2} \binom{2}{i}(\nabla_{s_1}^i\nabla_{s_2}^j \gp)^2dz_1dz_2$ is the roughness penalty and $\nabla_{s_h}^{v}\gp(\s)$ is the $v$th order derivative in the direction $\s_h,$ $h=1,2$.
With the bivariate spline representation $\gp(\s) = \B(\s)^\top\Valpha,$ we rewrite the objective function as:
\begin{align}\label{Eq: min_spline2}
\min_{\Valpha \in \Rc^{|\Mcc|}, \vbeta \in \Rc^{np}} 
\frac{1}{2n}\sum_{i=1}^{n}\big(y(\s_i)\! - \!\B(\s_i)^\top\Valpha \!-\! \x(\s_i)^\top\vbeta(\s_i)\big)^2\! +\! \rho \Valpha^\top \P \Valpha \!+\! \lambda \mathcal{P}(\vbeta),\text{s.t.}~ \K\Valpha = \0,
\end{align}
where $\P$ is the penalty matrix satisfying $\Valpha^\top \P \Valpha = \mathcal{J}(\gp),$ 
and $\K\Valpha = \0$ is the smoothness constraint of the bivariate splines across all the shared edges of triangles.

Then we can remove the constraint via a reparameterization on the spline coefficient $\Valpha$. Consider the QR decomposition of $\K^\top$: $\K^\top = \Q\R=
\begin{pmatrix}
\Q_1,\Q_2
\end{pmatrix}
\begin{pmatrix}
\R_1\\
\R_2
\end{pmatrix},$ where $\Q$ is orthogonal, $\R$ is upper triangular, the submatrix $\Q_1$ is the first $r$ columns of $\Q$, where $r$ is the rank of matrix $\K$, and $\R_2$ is a matrix of zeros. 
We reparametrize using $\Valpha = \Q_2\vpsi$ for some $\vpsi,$ then the smoothness constraint $\K\Valpha = \0$ is satisfied. 
Thus, the above problem can be rewritten as an unrestricted penalized minimization problem:
\begin{align}\label{Eq: min_spline3}
\min_{\vpsi,\vbeta} 
\frac{1}{2n}\sum_{i=1}^{n}\big(y(\s_i) - \Bt(\s_i)^\top\vpsi - \x(\s_i)^\top\vbeta(\s_i)\big)^2 + \rho \vpsi^\top \D \vpsi + \lambda \mathcal{P}(\vbeta),
\end{align}
where $\Bt(\s_i) = \Q_2^\top\B(\s_i)$, $\D = \Q_2^\top\P\Q_2$ and $\vbeta = [\vbeta(\s_1)^{\top},\cdots,\vbeta(\s_n)^{\top}]^{\top}.$
Let $\hat{\vbeta}$ and $\hat{\vpsi}$ be the minimizer of (\ref{Eq: min_spline3}). 
Then the estimator of $\Valpha$ and $g(\cdot)$ are $\hat{\Valpha} = \Q_2\hat{\vpsi}$ and $\hat{g}(\s) = \B(\s)^\top \hat{\Valpha}.$

\subsection{Forest Lasso Penalty}

In the following, we propose our design of the clustering penalty function $\mathcal{P}(\cdot).$
The modern spatial datasets obtained from remote sensing technologies usually contain thousands of spatial locations.
The massive volume of the data makes the spatial pairwise penalty function \citep{wang2019spatial} computationally infeasible.
Thus, we follow the idea of the tree-based Lasso penalty function \citep{li2019spatial,zhang2019distributed} to maintain the computation complexity.

Consider an undirected connected graph $\Gc = (\Vc,\Ec),$ where $\Vc = \{v_1,\cdots,v_n\}$ is the set of vertices with $v_i$ representing location $\s_i,$ and $\Ec = \{(v_i,v_j):v_i\neq v_j\}$ is the edge set.  
A spanning tree $\Tc$ of the graph $\Gc$ is a connected undirected subgraph of $\Gc$ with no cycles and includes all the vertices of $\Gc.$ 
With the $p$ specific spanning trees $\{\Tc_k\}_{k=1}^{p}$ corresponding to the $p$ covariates, the tree based Lasso penalty can be defined as:
\begin{align}
\Pc(\vbeta) = \sum_{k=1}^{p} \sum_{(\s_i,\s_j) \in \Tc_k} |[\vbeta(\s_i)]_k-[\vbeta(\s_i)]_k| = \sum_{k=1}^{p}\|\H_k\vbeta_k\|_1,
\end{align}
where $[\vbeta(\s_i)]_k$ presents the $k$th element of the vector $\vbeta(\s_i)$ and $(\s_i,\s_j) \in \Tc_k$ means there exists an edge in $\Tc_k$ connecting $\s_i$ and $\s_j,$ and $\H_k$ is the adjacency matrix of $\Tc_k.$

Note that there exists no circle in the spanning tree. 
Thus, for the tree $\Tc_k,$ its adjacency matrix $\H_k$ is full row rank with size $(n-1)\times n$. 
Adding one more row to $\H_k$, we can form a square and full rank matrix:
$\widetilde{\H}_k=\begin{bmatrix}
   \H_k\\
   \frac{1}{\sqrt{n}}\1^\top
\end{bmatrix},$ as in \cite{li2019spatial}.
Define $\vtheta = \widetilde{\H} \vbeta$ with $\widetilde{\H} = \Diag[\widetilde{\H}_1,\cdots,\widetilde{\H}_p].$
Then the estimation problem (\ref{Eq: min_original}) reduces to a lasso-type problem:
\begin{align}\label{Eq: min_clustering2}
\min_{g\in\mathbb{G},\vtheta \in \Rc^{nd}} \frac{1}{2n}\|\y - \g - \X \widetilde{\H}^{-1} \vtheta\|_2^2 + \rho \mathcal{J}(g) + \lambda \sum_{l \in \Ac} |[\vtheta]_l|,
\end{align}
where $\y=[y(\s_1),\cdots,y(\s_n)]^\top,$ $\g=[g(\s_1),\cdots,g(\s_n)]^\top,$ $\X=\Diag[\x(\s_1)^\top,\cdots,\x(\s_n)^\top]$, $\Ac$ presents the index set $\Ac = \{l: \text{mod}(l,n) \neq 0, \text{for~} l = 1, \cdots, nd \}$ and $[\vtheta]_l$ is the $l$th element of $\vtheta.$
With the BPST approximation, we have the estimation problem as:
\begin{align}\label{Eq: min_final}
\min_{\vpsi,\vtheta} 
\frac{1}{2n}\|\y - \Bt\vpsi - \widetilde{\X} \vtheta\|_2^2 + \rho \vpsi^\top \D \vpsi + \lambda \sum_{l \in \Ac} |[\vtheta]_l|,
\end{align}
where $\Bt = [\Bt(\s_1)^\top,\cdots, \Bt(\s_n)^\top]$ and $\widetilde{\X} = \X\widetilde{\H}^{-1}.$
The double penalized problem (\ref{Eq: min_final}) can be minimized with the R package \texttt{gelnet}.
By solving (\ref{Eq: min_final}), we can obtain the estimation of $\vtheta$ as $\hat{\vtheta}.$ Then, the estimator of $\vbeta$ is $\hat{\vbeta} = (\widetilde{\H})^{-1} \hat{\vtheta}.$ 

The model estimation from (\ref{Eq: min_final}) highly depends on the pre-selected tree sets $\{\Tc_k\}_{k=1}^{p}.$
In this work, we aim to improve estimation efficiency by exploiting various spanning tree structures.
Consider $Q$ trials, and for the $q$th trial, a set of spanning trees $\{\Tc_k^{(q)}\}_{k=1}^{p}$ is randomly generated.
The estimators $\hat{\vbeta}^{(q)}$ can be obtained by solving the corresponding estimation problem (\ref{Eq: min_final}).
After the $Q$ trials of estimations, the average of $\{\hat{\vbeta}^{(q)}\}_{q=1}^{Q}$ is calculated as $\bar{\vbeta} = \frac{1}{Q}  \sum_{q=1}^Q\hat{\vbeta}^{(q)}.$
Then for any two locations $\s_i$ and $\s_j,$ define the adaptive weight as $\w_{k}(\s_i,\s_j) = |[\bar{\vbeta}(s_i)]_k - [\bar{\vbeta}(\s_j)]_k|,$ $\forall k.$
A set of the minimum spanning trees $\{\Tc_k^*\}_{k=1}^{p}$ can be constructed with the adaptive weights, where $\Tc_k^*$ is for the $k$th covariate and based on $\{\w_k(\s_i,\s_j):i,j=1,\cdots,n\}$.
In the end, to obtain the final estimation $\hat{\g}^*$ and $\hat{\vbeta}^*$, we adopt the adaptive lasso penalty \citep{zou2006adaptive,huang2008adaptive}:
\begin{align}
\min_{\vpsi,\vbeta} 
\frac{1}{2n}\sum_{i=1}^{n}\big(y(\s_i) - \Bt(\s_i)^\top\vpsi - \x(\s_i)^\top&\vbeta(\s_i)\big)^2 + \rho \vpsi^\top \D \vpsi \notag\\
&+ \lambda  \sum_{k=1}^{p} \sum_{(\s_i,\s_j) \in \Tc_k^*} \frac{|[\vbeta(\s_i)]_k-[\vbeta(\s_i)]_k|}{\w_{k}(\s_i,\s_j)}.
\end{align}
Similarly, the above estimation problem can be rewritten in a matrix-vector form:
\begin{align}\label{Eq: Final_Estimator}
(\hat{\vpsi}^*,\hat{\vtheta}^*) = \arg \min_{\vpsi,\vtheta} 
\frac{1}{2n}\|\y - \Bt\vpsi - \widetilde{\X}^* \vtheta\|_2^2 + \rho \vpsi^\top \D \vpsi + \lambda \sum_{l \in \Ac} |[\vtheta]_l|/[\w]_l,
\end{align}
where $\widetilde{\X}^* = \X\widetilde{\H}^{*-1}$, $\widetilde{\H}^* = \Diag[\widetilde{\H}_1^*,\cdots,\widetilde{\H}_d^*],$ $\widetilde{\H}_k^*=\begin{bmatrix}
   \H_k^*\\
   \frac{1}{\sqrt{n}}\1^\top
\end{bmatrix}$ is the adjacency matrix corresponding to $\Tc^*_k$, and $\w = |\widetilde{\H}^{*}\bar{\vbeta}|$ is the adaptive weight vector.
With the estimated coefficient $(\hat{\vpsi}^*,\hat{\vtheta}^*)$, the estimated non-parametric function $\hat{g}(\s)^* = \widetilde{\B}(\s)\hat{\vpsi}^*$ and clustered linear coefficient is $\hat{\vbeta}=(\widetilde{\H}^*)^{-1}\hat{\vtheta}^*.$
Furthermore, $\s_i$ and $\s_j$ are treated from the same cluster if $\hat{\vbeta}^*(\s_i) = \hat{\vbeta}^*(\s_j).$
We name this method as forest Lasso penalty because multiple trees are adopted in our estimation procedure.
In Figure \ref{Fig: Forest Lasso Demo}, we visualize our forest Lasso approach with the scenario of one covariate.

\begin{figure}[t!]
\begin{center}
\begin{tabular}{@{}c@{}}
\includegraphics[width=1\textwidth]{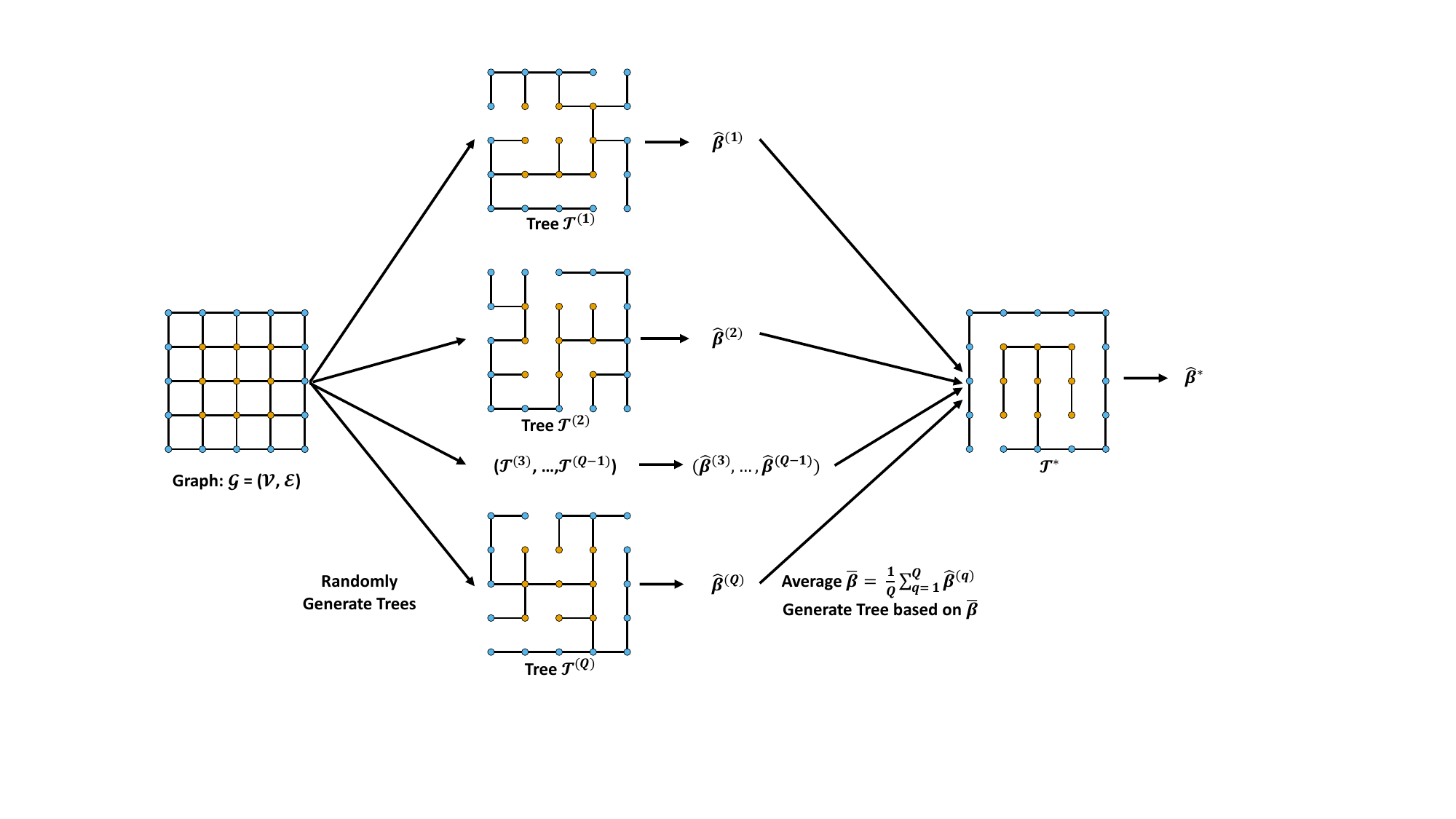}
\end{tabular}
\caption{Estimation procedure of the forest lasso. In this example, we consider a simple case with $\vbeta \in \Rc^1$ and two clusters distinguished by nodes' color (yellow and blue). } \label{Fig: Forest Lasso Demo}
\end{center}
\end{figure}

\subsection{Tuning Parameters Selection}

To select the tuning parameters $\rho$ and $\lambda$, we consider a two-step procedure to search from a sequence of grid points.
The idea is similar to the one proposed in \cite{zhu2018cluster}.
Recall that the parameter $\rho$ adjusts the smoothness of the BPST approximation, and $\lambda$ controls the homogeneity of the coefficients.
In general, we can implement a grid search for both the tuning parameters simultaneously. 
However, it leads to an intensive computation.
To address this problem, we propose a two-step procedure in which we first search for an optimal value of $\lambda$ under $\rho = 0$, then select $\rho$ given the optimal $\lambda$ from the first step. 
Specifically, we select $\lambda$ by minimizing the modified Bayesian information criterion \citep{wang2009shrinkage}
\begin{align}
\text{mBIC}_\lambda = \log\Big(\frac{1}{n}\sum_{i=1}^{n}\big(y(\s_i) - \hat{y}(\s_i)\big)^2\Big) + \log(\log(nd))\times \frac{\log(n)}{n}\times \text{df}(\hat{\vtheta}),
\end{align}
and $\rho$ is selected by minimizing the classical Bayesian information criterion
\begin{align}
\text{BIC}_\rho = \log\Big(\frac{1}{n}\sum_{i=1}^{n}\big(y(\s_i) - \hat{y}(\s_i)\big)^2\Big) + \frac{\log(n)}{n}\times \text{df}(\hat{\vpsi}),
\end{align}
where $\text{df}(\hat{\vtheta})$ is the total number of non-zero elements of $\hat{\vtheta}$ and $\text{df}(\hat{\vpsi}) = \tr[\Bt(\Bt^\top\Bt +\rho \D)^{-1}\Bt^\top].$


\section{Theoretical Properties}\label{Section: analysis}

In this section, we study the theoretical properties of the proposed doubly penalized least square estimator. 
Due to the high dimensionality of the design matrix $\widetilde{\X}^* \in \mathbb{R}^{n\times np}$, there is no closed-form solution to the minimization problem (\ref{Eq: Final_Estimator}).
Thus, we study the properties of our estimator $(\hat{g}^*,\widehat{\vtheta}^*)$ by introducing oracle estimator, which is defined as the followings.
Without loss of generality, we consider that only the first $p^\prime$ elements of $\vtheta$ are nonzero and the rest are all zeros, where $\vtheta$ is the true parameter vector defined by the adaptive tree. 
More specifically, we define non-zero index set as $\mathcal{N} = \{1, \ldots, p^\prime\}$ and zero index set as $\mathcal{N}^c = \{p^\prime+1, \ldots, np\}$. 
Denote $\widetilde{\X} = (\widetilde{\X}_{\mathcal{N}}, \widetilde{\X}_{\mathcal{N}^c})$ and $\widetilde{\X}_{\mathcal{N}}$ be the first $p^\prime$ columns of matrix $\widetilde{\X}$. 
Then the oracle estimator is defined as
\[
(\widehat{\bs{\psi}}^o, \widehat{\bs{\theta}}^o) =  \arg\min_{\vpsi^o,\vtheta^o} \frac{1}{2n}\|\y - \Bt\vpsi^o - \widetilde{\X}_{\mathcal{N}} \vtheta^o\|_2^2 + \rho \vpsi^{o \top} \D \vpsi^o.
\]
Define the oracle design matrix $\mathbf{E} = 
\begin{bmatrix}
\Bt, \widetilde{\X}_{\mathcal{N}}
\end{bmatrix}$ and the block diagonal penalty matrix $\widetilde{\D} = \begin{bmatrix}
   \D & \0\\
   \0 & \0
\end{bmatrix}$. 
It can be easily obtained that the oracle estimator $(\widehat{\vpsi}^{o\top}, \widehat{\vtheta}^{o\top})^{\top} = (\mathbf{E}^\top \mathbf{E} + n \rho \widetilde{\D})^{-1} \mathbf{E}^\top \y$, and $\hat{g}^{o}(\s) = \widetilde{\B}(\s)\widehat{\vpsi}^{o\top}$.
The following lemma gives the error bounds for the oracle estimator.

\begin{lem}\label{LEM:oracle}
Under Assumptions (A1) -- (A6) in the Appendix A, the oracle estimator satisfies the following properties, 
\begin{align}
&\| \widehat{\vtheta}^o - [\vtheta]_{\mathcal{N}}\|  = O_P\{|\triangle|^{(d+1)} + n^{-1/2} |\triangle|^{-1} + \rho |\triangle|^{-4}\},\\
&\|\hat{g}^o - g\|_{L_2} = O_P\{|\triangle|^{(d+1)} + n^{-1/2} |\triangle|^{-1} + \rho |\triangle|^{-4}\},
\end{align}
where $|\triangle|:=\max \{|\tau|,\tau \in \triangle \}$ is the size of $\triangle$, and $\|f\|_{L_2}=\int_{\s\in \Omega} f(\s)^2d(\s)$ for any function $f$.
\end{lem}

Lemma \ref{LEM:oracle} shows that the estimation errors of the oracle estimator depend on the triangle size $|\triangle|$, data size $n$ and penalty parameter $\rho$.
By properly choosing the size with $|\triangle|^{-1} = o(\sqrt{n})$ and penalty parameter with $\rho = o(|\triangle|^{4})$, we have the consistent estimator with $n \to \infty$.
Next, we show in Theorem \ref{THE:consistency} that under some mild conditions, our proposed estimator is also consistent. 
The key idea of the analysis is that the oracle estimator satisfies the KKT condition of the problem (\ref{Eq: Final_Estimator}) almost surely. 
Then our proposed estimator is the same as the oracle one and has consistency.

\begin{thm}\label{THE:consistency}
Under Assumptions (A1) -- (A7) in the Appendix A, our proposed estimator has the following error bounds,
\begin{align}
&\| \widehat{\vtheta}^* - \vtheta \| = O_P\{|\triangle|^{(d+1)} + n^{-1/2} |\triangle|^{-1} + \rho |\triangle|^{-4} + \lambda |\mathcal{N}\cap \mathcal{A}|^{1/2}\},\label{Eq: them1_1}\\
&\|\hat{g}^* - g\|_{L_2} = O_P\{|\triangle|^{(d+1)} + n^{-1/2} |\triangle|^{-1} + \rho |\triangle|^{-4} + \lambda |\mathcal{N}\cap \mathcal{A}|^{1/2}\} \label{Eq: them1_2}.
\end{align}
Furthermore, our estimator $\widehat{\vtheta}^*$ has the selection consistency that, as $n \to \infty$,
\begin{align}
\mathbb{P}\{ \widehat{\mathcal{N}} = {\mathcal{N}} \} \to  1,
\end{align}
where $ \widehat{\mathcal{N}} = \{i, [\widehat{\vtheta}^*]_i \neq 0\}$ is the estimated non-zero index set of $\widehat{\vtheta}^*$. 
\end{thm}

Theorem \ref{THE:consistency} (\ref{Eq: them1_1})-(\ref{Eq: them1_2}) shows that besides the error terms of the oracle estimator, the error bounds are also affected by the tuning parameter $\lambda$ and the term $|\mathcal{N}\cap \mathcal{A}|$. 
Recall that $|\mathcal{N}|$ is the number of nonzero elements in the true linear coefficient $\vtheta$, and $\mathcal{A}$ is the index set of the penalized elements.
By having the same settings as the oracle estimator and further setting $\lambda = o(|\mathcal{N}\cap \mathcal{A}|)$, we have the estimation consistency of the proposed estimator.
Additionally, the error term $O_P(\lambda|\mathcal{N}\cap \mathcal{A}|)$ indicates that a tree with less inter-cluster edges will lead to a better estimation accuracy.
This further explains why the adaptive tree structure performs better than a randomly generated tree.


\section{Simulation Study}\label{Section: Simulation}

In this section, we present simulation studies to show the performance of our method under different scenarios.
The computations were performed on the linux server\footnote{https://researchit.las.iastate.edu/biocrunch}.
Totally 1000 locations are randomly generated in a square domain $[0,1]\times [0,1]$.
The underlying model is 
\begin{align}
y(\s_i) = g(\s_i) + x_1(\s_i)\beta_1(\s_i) + x_2(\s_i) \beta_2(\s_i) + \epsilon(\s_i), ~ \epsilon(\s_i) \stackrel{i.i.d.}{\sim} N(0,\sigma^2).
\end{align}
The variance of the random error $\epsilon(\s_i)$ is set as $\sigma^2=0.1$.
The true coefficients are to be a constant within each cluster and the numbers of clusters are $4$ for $\beta_1$, $5$ for $\beta_2$ (see Figure \ref{Fig: True Parameters}).
The nonlinear term $g(\s_i)$ is generated from the spatial Gaussian process with zero mean and an isotropic exponential covariance function as $C(g(\s_1),g(\s_2)) = \exp(-\|\s_1-\s_2\|/r)$, where $r$ is the range parameter.
In our numerical study, we consider two different settings: weak correlation with $r = 1,$ and strong correlation with $r = 10.$ The generated nonlinear surfaces are shown in Figure \ref{Fig: True g}.

\begin{figure}[t!]
\begin{center}
\begin{tabular}{@{}cc@{}}
\includegraphics[width=.4\textwidth]{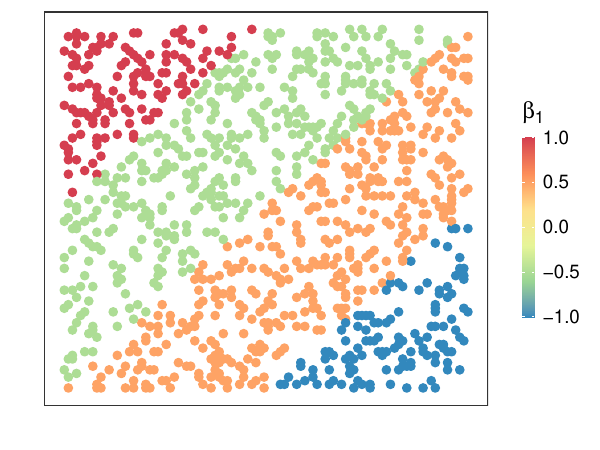}&
\includegraphics[width=.4\textwidth]{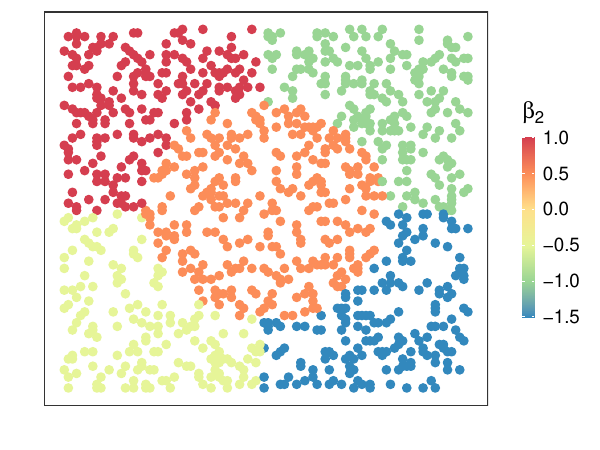}
\end{tabular}
\caption{Spatial structures of true coefficients $\beta_1$ and $\beta_2$} \label{Fig: True Parameters}
\end{center}
\end{figure}

\begin{figure}[t!]
\begin{center}
\begin{tabular}{@{}cc@{}}
\includegraphics[width=.4\textwidth]{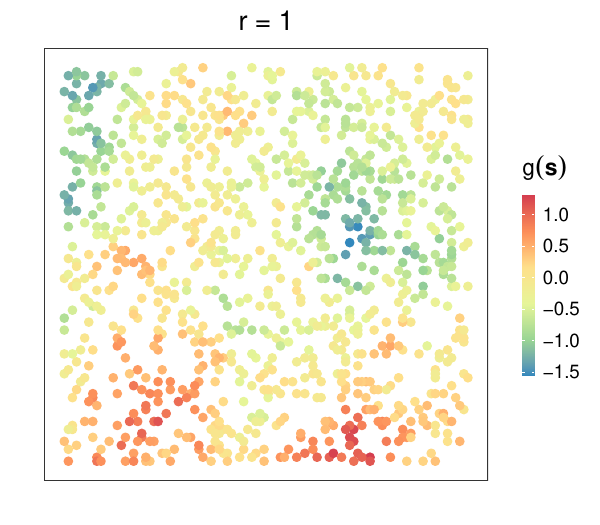}&
\includegraphics[width=.4\textwidth]{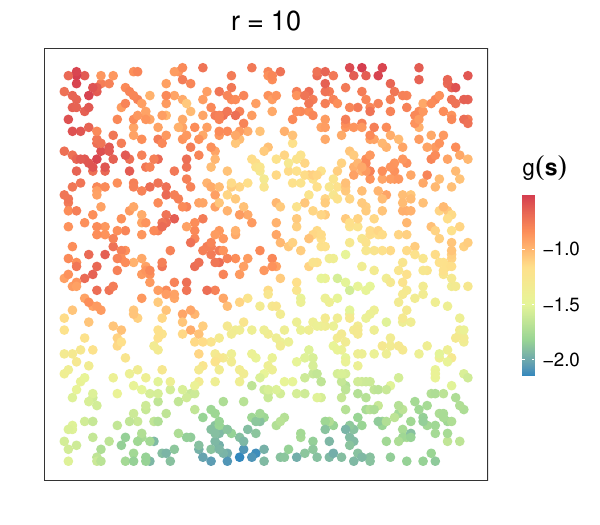}
\end{tabular}
\caption{Spatial varying nonlinear function $g(\s)$ under different range parameters $r$} \label{Fig: True g}
\end{center}
\end{figure}

In practice, the covariates appear to a correlation with some other covariates as well as a spatial correlation among the location.
Thus, both the spatial dependence and collinearity are introduced among $\{x_1(\s_i),x_2(\s_i)\}_{i=1}^{n}.$
Specifically, two independent realizations $\{z_1(\s_1)\}_{i=1}^{n}$ and $\{z_2(\s_1)\}_{i=1}^{n}$ are obtained from the spatial Gaussian process with zero mean and an isotropic exponential covariance function as $C(z_k(\s_i),z_k(\s_j)) = \exp(-\|\s_i-\s_j\|/r^\prime),$ $k=1,2.$
Then, $x_1(\s_i)$ and $x_2(\s_i)$ are generated with $x_1(\s_i) = z_1(\s_i)$ and $x_2(\s_i) = c z_1(\s_i) + \sqrt{1-c^2} z_2(\s_i).$
In our simulation, we use the case with $c=0.75$ and $r^\prime = 0.1$.

We compare our SHAPLM with the geographically weighted regression (GWR) model and the parametric spatial clustered coefficient model (PSCCM).
In the estimation of PSCCM, the fusion penalty is adopted for both the intercept and the covariate coefficients:
$\min \frac{1}{2n}\sum_{i=1}^{n}\big(y(\s_i) - \beta_0(\s_i) - \x(\s_i)^\top\vbeta(\s_i)\big)^2 + \lambda \mathcal{P}(\vbeta).$
The key difference between PSCCM and SHAPLM is that PSCCM uses the spatially piece-wise constant to fit the intercept while SHAPLM adopts the spatially smoothing intercept.
For both the PSCCM and SHAPLM, we consider the number of the trees $Q$, i.e. the trials of estimation, as $\{0,5,10,20\}.$
When the tree number is $0,$ it is corresponding to the minimum spanning tree penalty for the final estimation.
Note that the PSCCM with $Q=0$ is the method proposed in \cite{li2019spatial}.

The following criteria are adopted to quantify the estimation efficiency:
1) the mean squared error (MSE) for the model estimation
\begin{align}
\text{MSE}(\vbeta) = \frac{1}{n} \sum_{i=1}^{n} \|\vbeta(\s_i)-\hat{\vbeta}^*(\s_i)\|^2~ 
\text{and}~ 
\text{MSE}(g) = \frac{1}{n} \sum_{i=1}^{n} \big(g(\s_i)-\hat{g}^*(\s_i)\big)^2;
\end{align}
2) the Rand index (RI) \citep{rand1971objective} for the partition recovery, which is the percentage of correct partition in the estimation.
The value of RI lies between 0 and 1. The higher RI indicates better performance of the method.
The simulation results are shown in Table \ref{Table: Simu1}.

First of all, it can be seen that for SHAPLM, as the tree number increases, the estimation performances get improved. 
For example, using the forest Lasso penalty with 10 trees, the SHAPLM reduced MSE of $\vbeta$ by around 75\% and MSE of $g(\cdot)$ by about 85\%, compared with the minimum spanning tree penalty method ($Q =0$) and under the weak correlation setting.
By comparing the results of PSCCM and SHAPLM, we can see that the estimation performance of PSCCM is not stable and worse. This is because the PSCCM is using a spatial piece-wise constant surface to approximate the nonlinear function, which causes the incorrect adaptive weights.
Additionally, the SHAPLM has a better estimation accuracy as the nonparametric function $g(\cdot)$ gets more smooth, i.e. the stronger correlation. 
This is from the fact that when the nonparametric part becomes more smooth, the estimation efficiency of the bivariate spline method gets improved, which in turn improves the estimation performance of the parametric part.
Additionally, compared with the estimation results of GWR, the MSE of SHAPLM is significantly reduced with $Q \ge 1$.
This performance gain is because the clustered constant structure of $\vbeta$ breaks the spatial smoothly varying coefficient assumption for GWR.

\begin{table}[t!]
\centering
{\small\caption{Estimation results based on BIC criterion}\label{Table: Simu1}
\vspace{0.1in}
\begin{tabular}{lc|cccc|cccc}
  \hline
  \hline
  \multirow{2}{*}{Method }& \multirow{2}{*}{\#Tree} & \multicolumn{4}{c|}{Weak Correlation $r = 1$} & \multicolumn{4}{c}{Strong Correlation $r = 10$}\\
 & & $\text{MSE}(\vbeta)$ & $\text{MSE}(g)$ & $\text{RI}(\beta_1)$ & $\text{RI}(\beta_2)$ & $\text{MSE}(\vbeta)$ & $\text{MSE}(g)$ & $\text{RI}(\beta_1)$ & $\text{RI}(\beta_2)$ \\ 
 \hline
 GWR & - & 0.230 & 0.158 & - & - & 0.199 & 0.128 & - & - \\ 
 \hline
\multirow{4}{*}{PSCCM}
& 0 & 0.246 & 0.166 & 0.64 & 0.75  & 0.624 & 1.496 & 0.58 & 0.66 \\  
& 5 & 0.220 & 0.179 & 0.65 & 0.79 & 0.661 & 1.531 & 0.57 & 0.67 \\ 
&10 & 0.238 & 0.190 & 0.62 & 0.78 &  0.656 & 1.522 & 0.57 & 0.67 \\ 
& 20 & 0.237 & 0.197 & 0.61 & 0.79  & 0.813 & 2.963 & 0.54 & 0.67 \\ 
    \hline 
\multirow{4}{*}{SHAPLM} 
& 0 &  0.234 & 0.296 & 0.64 & 0.76  & 0.160 & 0.143 & 0.71 & 0.81 \\ 
& 5 &  0.067 & 0.055 & 0.86  & 0.91 & 0.047 & 0.012 & 0.91 & 0.94 \\ 
& 10 & 0.059 &  0.049 &   0.87 &  0.91 & 0.041 & 0.012 & 0.92 & 0.94 \\ 
& 20 & 0.054 & 0.045 & 0.88 & 0.92 & 0.034 &  0.011  & 0.93 & 0.95 \\ 
   \hline
   \hline
\end{tabular}}
\end{table}

Next, we examine the computation and estimation performances of our proposed forest Lasso penalty by comparing it with the graph Lasso penalty \citep{hallac2015network}.
For the graph Lasso penalty, we generate the initial graph $\Gc$ with the Delaunay triangulation and formulate the fusion penalty as $\sum_{k=1}^{ {p}} \sum_{(\s_i,\s_j) \in \Gc}|[\vbeta(\s_i)]_k-[\vbeta(\s_i)]_k|$.
We apply the two methods to the dataset with strong correlation $r=10$ and compare the performances with computation time and mean square errors, MSE$(\vbeta)$ and MSE$(\g)$.
The results are shown in Figure \ref{Fig: SHAPLM_SASA}.
It can be seen that though the graph Lasso has a better estimation accuracy than the forest Lasso with $Q=0$ (i.e. the distance-based MST Lasso \citep{li2019spatial}), the forest Lasso outperforms the graph Lasso in both MSE$(\vbeta)$ and MSE$(\g)$ when $Q \ge 3$.  
Meanwhile, our forest Lasso gains a significant improvement in terms of computation efficiency.
The graph Lasso requires around 50 min to finish the computation.
However, to achieve the same accuracy, our forest Lasso only takes less than 10 min (when $Q=3$).
Thus, the proposed forest Lasso is much more efficient than the traditional graph Lasso penalty.

\begin{figure}[t!]
\begin{center}
\begin{tabular}{@{}c@{}}
\includegraphics[width=1\textwidth]{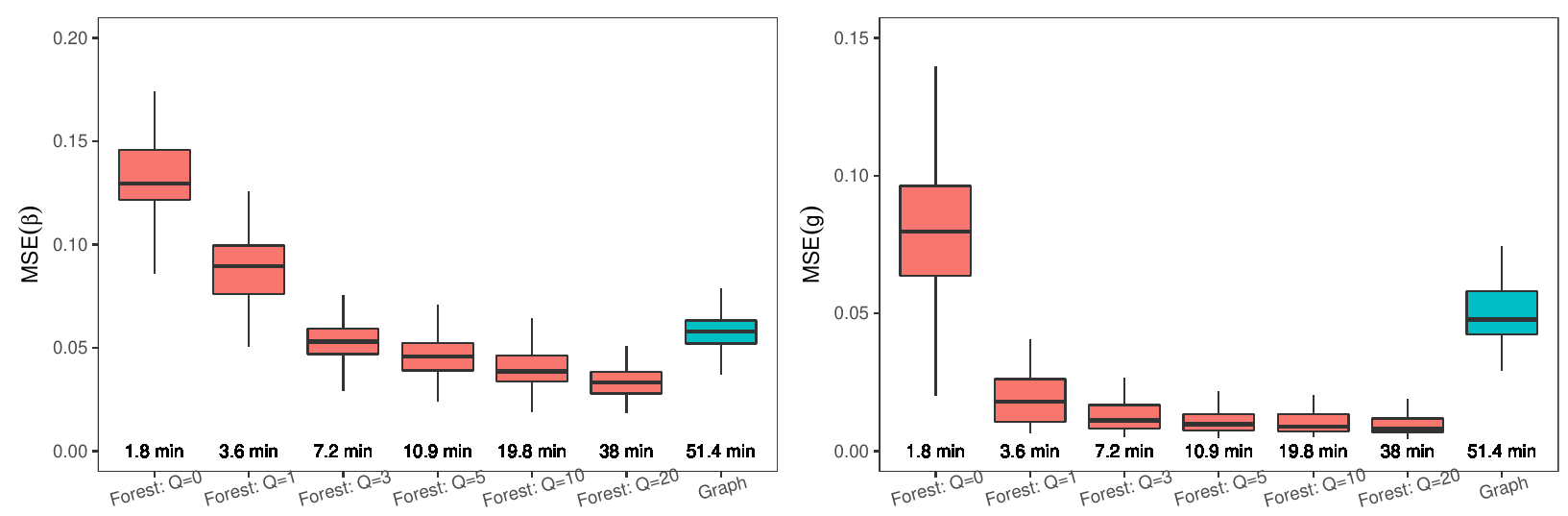}
\end{tabular}
\caption{The boxplots of MSE for the graph Lasso method and forest Lasso method with different number of trees $Q$. Computation times are reported above each method’s label.} \label{Fig: SHAPLM_SASA}
\end{center}
\end{figure}


\section{Temperature Data Analysis}\label{Section: RealData}

Air temperature (AT) is a key indicator when studying the impact of climate change on people, ecosystems, and energy system \citep{li2018developing}. 
However, high-resolution air temperature data are usually not widely available.
On the one hand, the air temperature data can only be measured and collected by some specific weather stations. 
This makes the data only available for a few spatial locations.
On the other hand, the air temperature data usually exhibit a complex spatial heterogeneity, due to the complex biophysical and socioeconomic factors.
Thus, we can not estimate the data with simple interpolation.
Fortunately, thanks to the advanced remote sensing technology, another data source, the land surface temperature (LST), can be obtained from the high-resolution images.
Consider these two temperatures are highly correlated. Hence, we can use the land surface temperature to predict air temperature.

We collect the air temperature data from weather stations of the Global Historical Climatology Network-Daily (GHCN-D) dataset (version 3), reported by the National Oceanic and Atmospheric Administration (NOAA) National Climate Data Center (NCNC) \citep{menne2012global,menne2012overview}. 
The GHCN-D provides meteorological observations from more than 90,000 land-based weather stations all over the world.
In our study, we focus on the stations in the conterminous United States and the air temperature record in 2010.
We adopt the gap-filled MODIS daily land surface temperature (LST) data \citep{li2018creating} as the auxiliary covariate.
The gap-filled MODIS daily LST data was produced by applying the three-step hybrid gap-filled approach (i.e. daily merging, spatio-temporal gap-filling, and temporal interpolation) to the MODIS LST daily observations, MYD11A1/MOD11A1 from Terra and Aqua satellites.

We apply our proposed SHAPLM to study the relationship between the daily maximum air temperature and land surface temperature.
The following spatial partial linear model is adopted,
\begin{equation}\label{Eq: AT_model}
y(\s_i) = g(\s_i) + x_{1}(\s_i)\beta_1(\s_i)  + \epsilon(\s_i),
\end{equation}
in which $y(\s_i)$, $x_1(\s_i)$ and $\epsilon(\s_i)$ are the daily maximum air temperature, the daily land surface temperature, and the measurement error at the $i$th station with the location $\s_i$, respectively.
We aim to identify spatial heterogeneous patterns based on the coefficient $\beta_1$.
In our implementation, the initial graph $\Gc$ is generated with the Delaunay triangulation and the forest lasso adopts $Q=10$. 

\subsection{Date 2010-12-21, Winter Solstice}

We first study the data of 2010-12-21, the winter solstice of the year. 
By removing the stations without data records, we have 4,637 remained stations.
The air temperature and land surface temperature are shown in Figure \ref{Fig: AT_LST_1221}.
From Figure \ref{Fig: AT_LST_1221} (a), it can be seen that the air temperature of the day ranges from $-20$ to $30$ degrees Celsius , the southern part of the United States has a higher temperature and the northern part has a lower temperature. 
The stations with temperatures above $20$ degrees Celsius  mainly concentrate around the state of Texas,  while those with extremely low temperatures (i.e. below $-10$ degrees Celsius) are in the state of Montana and North Dakota.
For the land surface temperature in Figure \ref{Fig: AT_LST_1221} (b), the basic pattern is similar to that of air temperature, while the region with extremely low temperature doesn't appear in the state of Montana.

\begin{figure}[!ht]
    \begin{minipage}{0.45\textwidth}
    \centering
        \begin{tabular}{c}  
        \includegraphics[width = 1\textwidth]{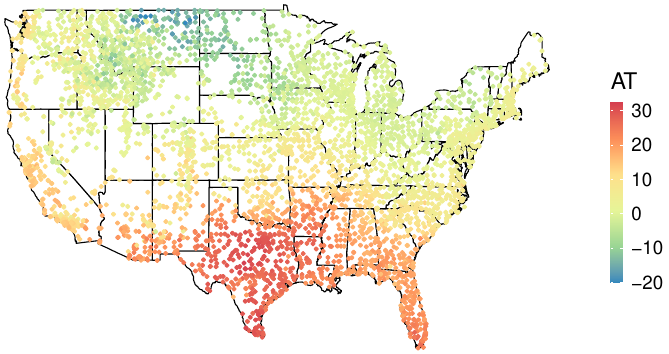}\\
        (a) Air temperature
        \end{tabular}
    \end{minipage}
    \hspace{0.1in}
    \begin{minipage}{0.45\textwidth}
    \centering
        \begin{tabular}{c}  
        \includegraphics[width = 1\textwidth]{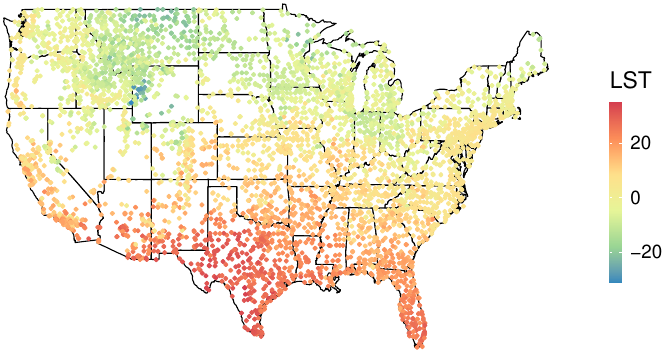}\\
        (b) Land surface temperature
        \end{tabular}
    \end{minipage}
    \caption{The temperature data of 2010-12-21.}\label{Fig: AT_LST_1221}
\end{figure}

We apply our SHAPLM to fit the data. The model estimations are shown in Figure \ref{Fig: SHAPLM_1221}.
It can be seen from Figure \ref{Fig: SHAPLM_1221} (a) that the estimated intercept $g$ range from $-10$ to about $25$, which explains the majority of the spatial variation.
The states of North Dakota and South Dakota have low values of estimated intercepts, while the states of Texas, Louisiana, and Mississippi have higher estimated intercepts.
From Figure \ref{Fig: SHAPLM_1221} (b), we can see that the United States is mainly divided into two parts according to the estimated coefficient $\hat{\beta}_1$, the eastern part with $\hat{\beta}_1 \le 0.3$ and the western part vice versa.
Furthermore, it is identified as a region in the northern part of Montana with $\hat{\beta}_1$ higher than $0.6$.
By checking with the temperature data, we find that this region has an extremely low air temperature around $-20$ degrees Celsius.

\begin{figure}[!ht]
    \begin{minipage}{0.45\textwidth}
    \centering
        \begin{tabular}{c}  
        \includegraphics[width = 1\textwidth]{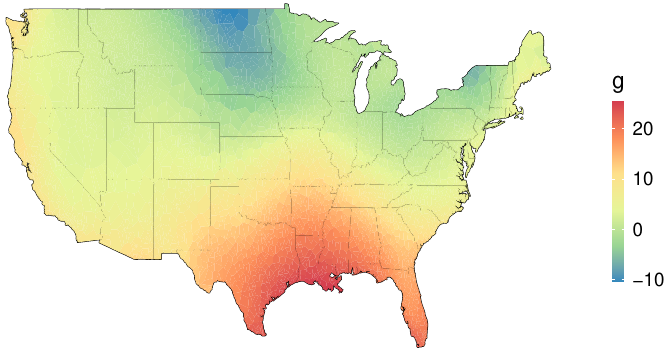}\\
        (a) Estimated intercept of SHAPLM
        \end{tabular}
    \end{minipage}
    \hspace{0.1in}
    \begin{minipage}{0.45\textwidth}
    \centering
        \begin{tabular}{c}  
        \includegraphics[width = 1\textwidth]{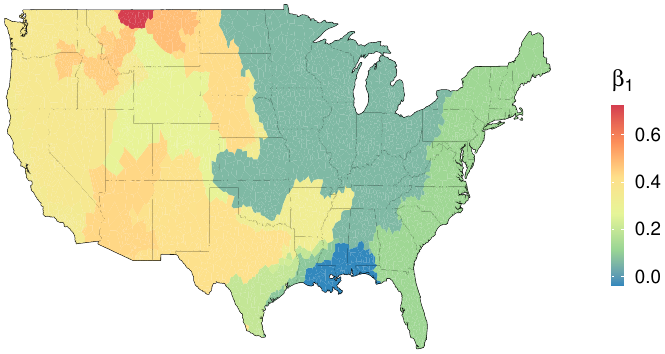}\\
        (b) Estimated coefficient of SHAPLM
        \end{tabular}
    \end{minipage}
    \caption{The SHAPLM estimation on data of 2010-12-21.}\label{Fig: SHAPLM_1221}
\end{figure}

We also analyze the data with the methods of GWR and PSCCM.
The fitting results are shown in Figure \ref{Fig: GWR_PSCCM_1221}.
From Figure \ref{Fig: GWR_PSCCM_1221} (a), the estimated intercept of GWR has a similar pattern as that of SHAPLM: the low values concentrate in the states of North Dakota and South Dakota, while the high values are in the states of Texas and Louisiana.
But SHAPLM has a smoother estimated intercept than GWR.
For the estimated coefficient, it can be seen from Figure \ref{Fig: GWR_PSCCM_1221} (b) that there also exists a region with large $\hat{\beta}_1$ in the northern part of Montana, which is consistent with the estimation of SHAPLM.
However, the estimation of GWR has a few outliers. 
For example, there is a small area in the northern part of Texas that has the value of $\hat{\beta}_1$ below $-1$, and another region at the border of North Dakota and South Dakota with $\hat{\beta}_1$ higher $1.5$.
It is hard to properly interpret these outlier regions. 

For the estimation of PSCCM shown in Figure \ref{Fig: GWR_PSCCM_1221} (c)-(d), the fitting results are not consistent with those of GWR and SHAPLM: the estimated intercept ranges from around $2.4$ to $3.6$, and the estimated coefficient can not reveal the heterogeneous patterns in Montana and Northern part of the United States, which are identified by SHAPLM and GWR.
Moreover, PSCCM detects several small clusters based on the estimated coefficient $\hat{\beta}_1$, which are marked with red circles.
It can be found that the small clusters in circle $1$ and $2$ are both appeared in the estimated intercept and coefficient, which indicates that the cluster detection in PSCCM is mutually affected.

\begin{figure}[!ht]
    \begin{minipage}{0.45\textwidth}
    \centering
        \begin{tabular}{c}  
        \includegraphics[width = 1\textwidth]{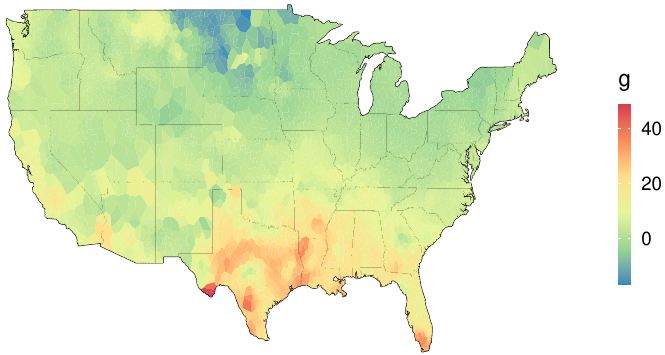}\\
        (a) Estimated intercept of GWR
        \end{tabular}
    \end{minipage}
    \hspace{0.1in}
    \begin{minipage}{0.45\textwidth}
    \centering
        \begin{tabular}{c}  
        \includegraphics[width = 1\textwidth]{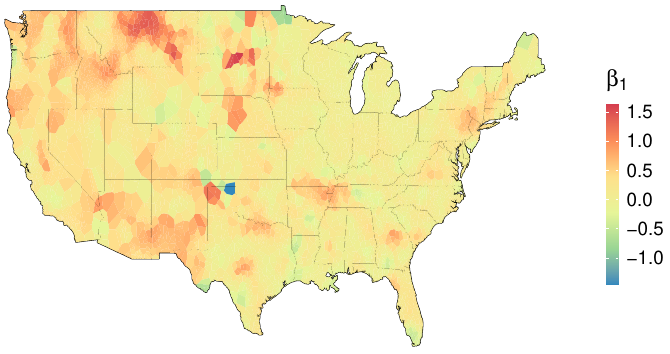}\\
        (b) Estimated coefficient of GWR
        \end{tabular}
    \end{minipage}\\
    \begin{minipage}{0.43\textwidth}
    \vspace{.2in}
    \centering
        \begin{tabular}{c}  
        \includegraphics[width = 1\textwidth]{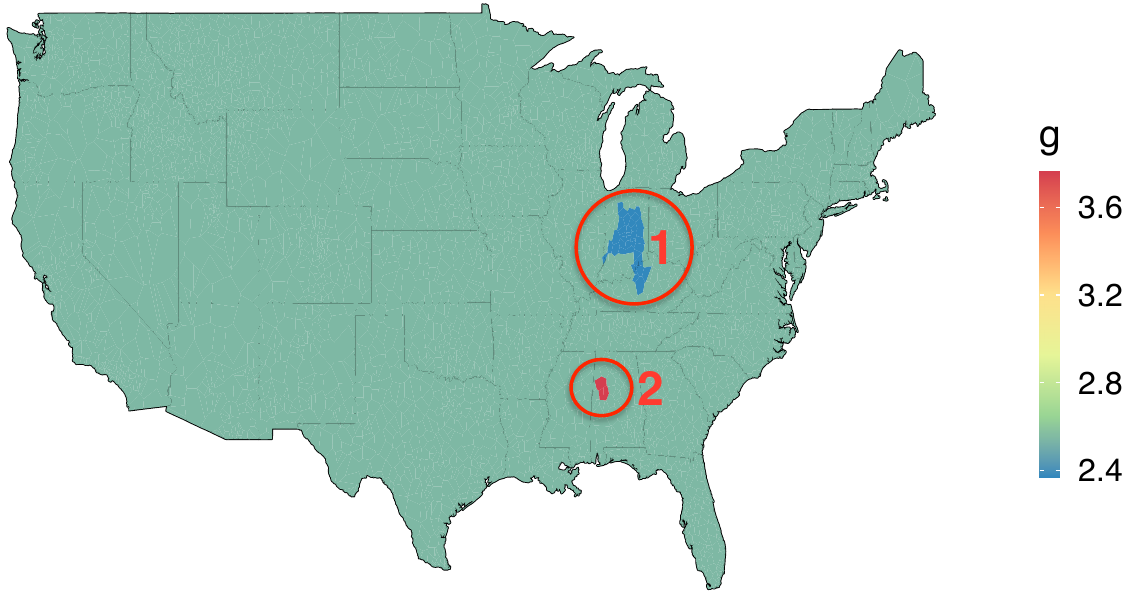}\\
        (c) Estimated intercept of PSCCM
        \end{tabular}
    \end{minipage}
    \hspace{0.3in}
    \begin{minipage}{0.43\textwidth}
    \vspace{.2in}
    \centering
        \begin{tabular}{c}  
        \includegraphics[width = 1\textwidth]{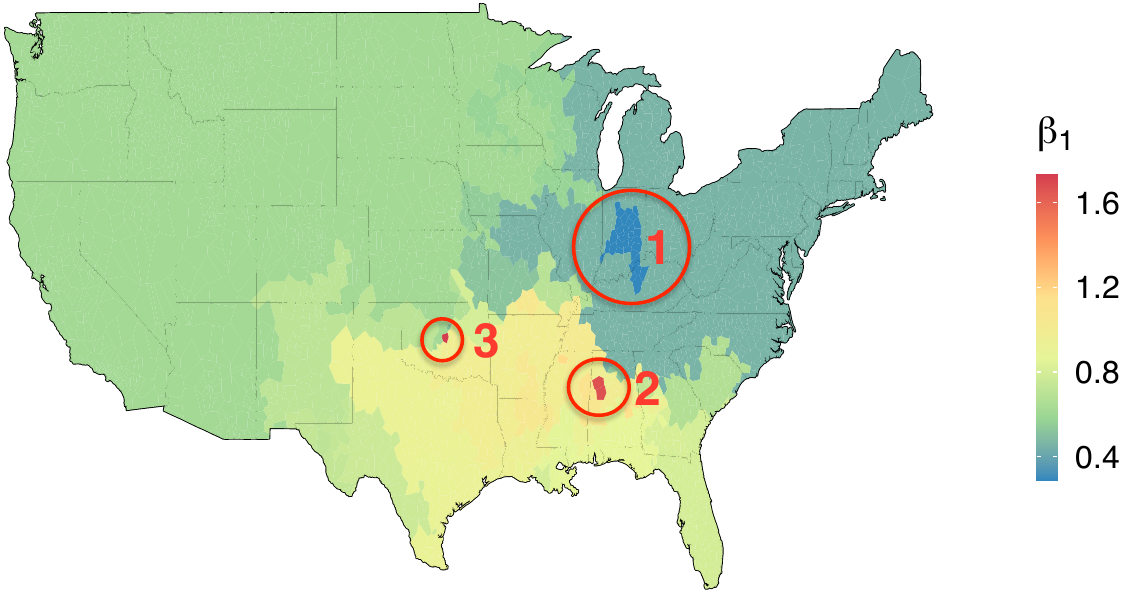}\\
        (d) Estimated coefficient of PSCCM
        \end{tabular}
    \end{minipage}
    \caption{The estimations of GWR and PSCCM on data of 2010-12-21.}\label{Fig: GWR_PSCCM_1221}
\end{figure}

\subsection{Date 2010-06-21, Summer Solstice}

We also study the data of the summer solstice on 2010-06-21. 
Totally 7001 stations have the temperature records, and the temperature data are shown in Figure \ref{Fig: AT_LST_0621}.
It can be seen that temperature patterns of the summer period are more complex than those of wintertime.
In Figure \ref{Fig: AT_LST_0621} (a), we find that the air temperature ranges from about $10$ to $40$ degrees Celsius.
The southern United States has a high air temperature of about $40$ degrees Celsius, while the northwestern United States keeps a low air temperature around $10$ degrees Celsius. 
While for the land surface temperature shown in \ref{Fig: AT_LST_0621} (b), we find that the temperature is higher than $60$ degrees Celsius in some regions of the states of California, Arizona, and New Mexico.
This might be because these regions are deserts.
Also, both of the temperatures show a stable pattern in the eastern part of the United States, while a strong heterogeneity in the western United States. 
We think this is because the topographic structures in the western United States are more complex than the eastern part.

\begin{figure}[!ht]
    \begin{minipage}{0.45\textwidth}
    \centering
        \begin{tabular}{c}  
        \includegraphics[width = 1\textwidth]{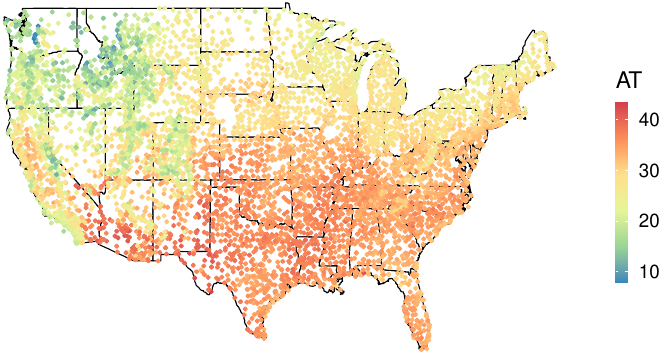}\\
        (a) Air temperature
        \end{tabular}
    \end{minipage}
    \hspace{0.1in}
    \begin{minipage}{0.45\textwidth}
    \centering
        \begin{tabular}{c}  
        \includegraphics[width = 1\textwidth]{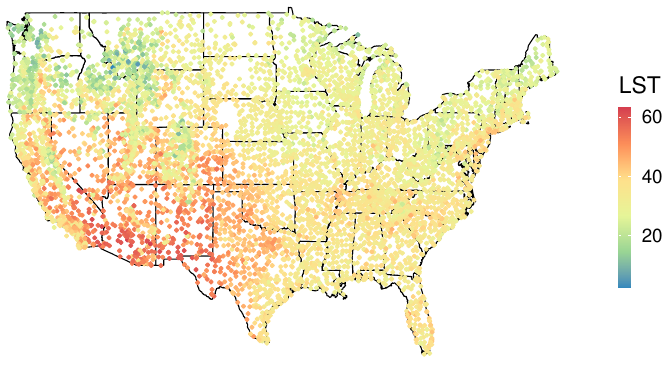}\\
        (b) Land surface temperature
        \end{tabular}
    \end{minipage}
    \caption{The temperature data of 2010-06-21.}\label{Fig: AT_LST_0621}
\end{figure}

Figure \ref{Fig: SHAPLM_0621} shows the estimation results of SHAPLM on the data.
The estimated intercept follows the pattern of the air temperature: The eastern part of the United States has high values while the western part is with lower values. 
The estimated coefficient $\hat{\beta}_1$ reveals the spatial heterogeneous cluster patterns.
We find that majority of the detected clusters are in the southwestern United States.
This is because of the complex temperature pattern in this area.
For example, there is a region with $\hat{\beta}_1$ smaller than $0.25$ near the Sierra Nevada.
This region is corresponding to the area with local low air temperature, around $10$ degrees Celsius.
Also, another region with $\hat{\beta}_1 \ge 0.40$ is detected near the Sonoran Desert.
This region has an extremely high land surface temperature, which is above $60$ degrees Celsius.
These cluster findings imply that the temperature relationship might be affected by the topographic structures.

\begin{figure}[!ht]
    \begin{minipage}{0.45\textwidth}
    \centering
        \begin{tabular}{c}  
        \includegraphics[width = 1\textwidth]{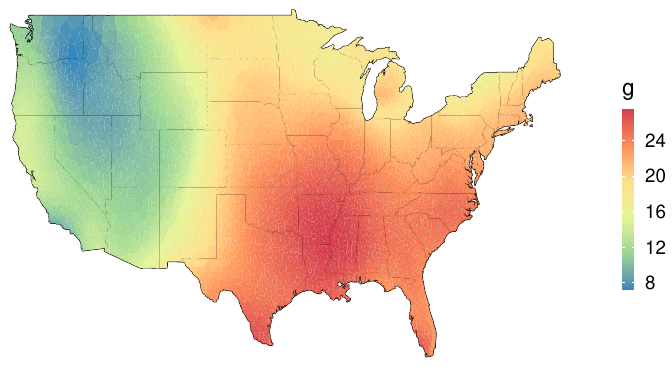}\\
        (a) Estimated intercept of SHAPLM
        \end{tabular}
    \end{minipage}
    \hspace{0.1in}
    \begin{minipage}{0.45\textwidth}
    \centering
        \begin{tabular}{c}  
        \includegraphics[width = 1\textwidth]{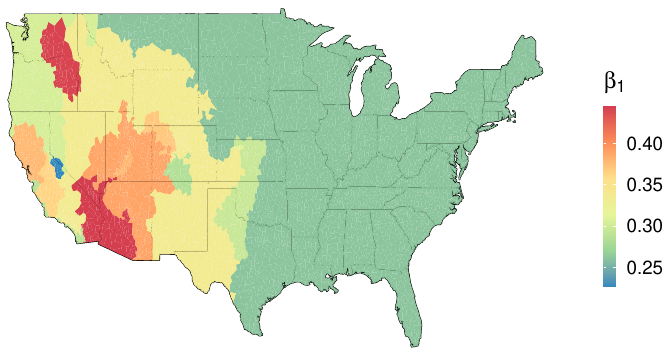}\\
        (b) Estimated coefficient of SHAPLM
        \end{tabular}
    \end{minipage}
    \caption{The SHAPLM estimation on data of 2010-06-21.}\label{Fig: SHAPLM_0621}
\end{figure}

We also apply GWR and PSCCM to the data (see Figure \ref{Fig: GWR_PSCCM_0621}), while less information can be obtained from the estimation results.
For GWR estimation, there are some outliers marked in the circle $1$.
These outliers have the estimated intercepts both below $0$ and above $75$, as well as the estimated coefficients both near $-3$ and $1$.
By checking the station map shown in Figure \ref{Fig: AT_LST_0621}, this might be caused by the lack of observations at the boundary.
For the estimation of PSCCM, the intercept has the same estimated value of about $11$ over the spatial domain.
This estimated result is inconsistent with the estimations of SHAPLM and GWR.
The estimated coefficient $\hat{\beta}_1$ of PSCCM has similar patterns as the estimated intercept of SHAPLM and GWR: low values in the northwestern part and high values in the southeastern part.
We think this is because the homogeneous estimated intercept makes the coefficient $\hat{\beta}_1$ explain the spatial variation in PSCCM.

\begin{figure}[!ht]
    \begin{minipage}{0.45\textwidth}
    \centering
        \begin{tabular}{c}  
        \includegraphics[width = 1\textwidth]{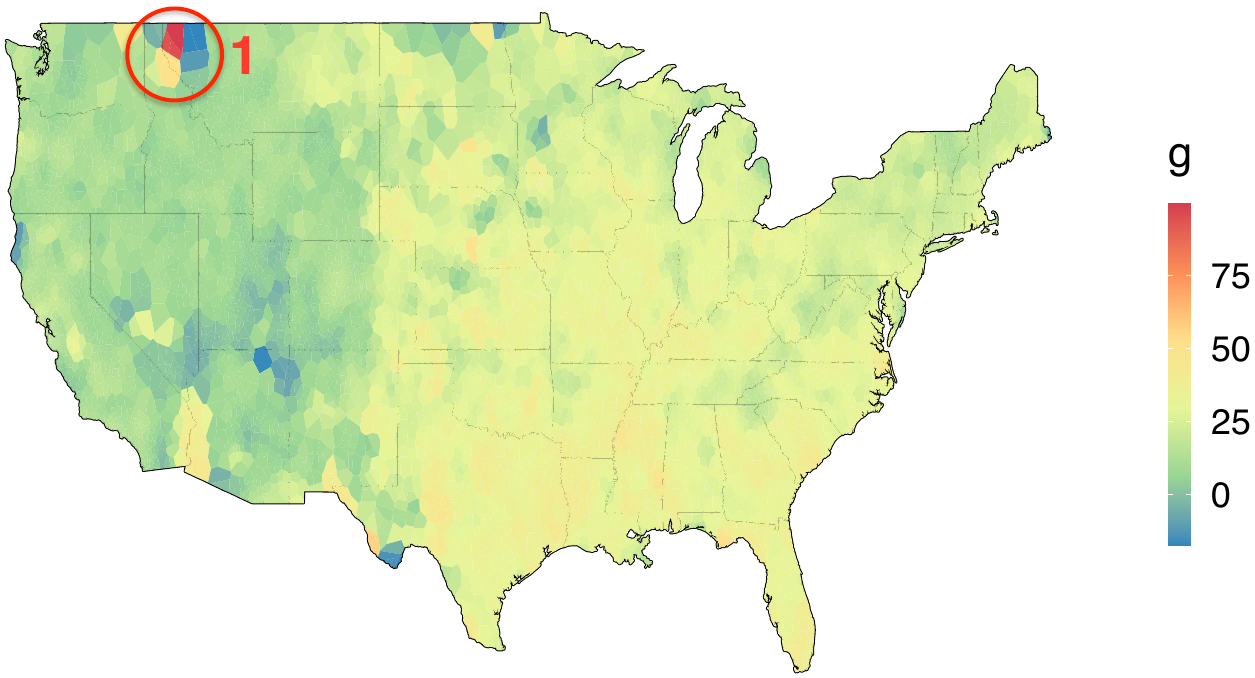}\\
        (a) Estimated intercept of GWR
        \end{tabular}
    \end{minipage}
    \hspace{0.1in}
    \begin{minipage}{0.45\textwidth}
    \centering
        \begin{tabular}{c}  
        \includegraphics[width = 1\textwidth]{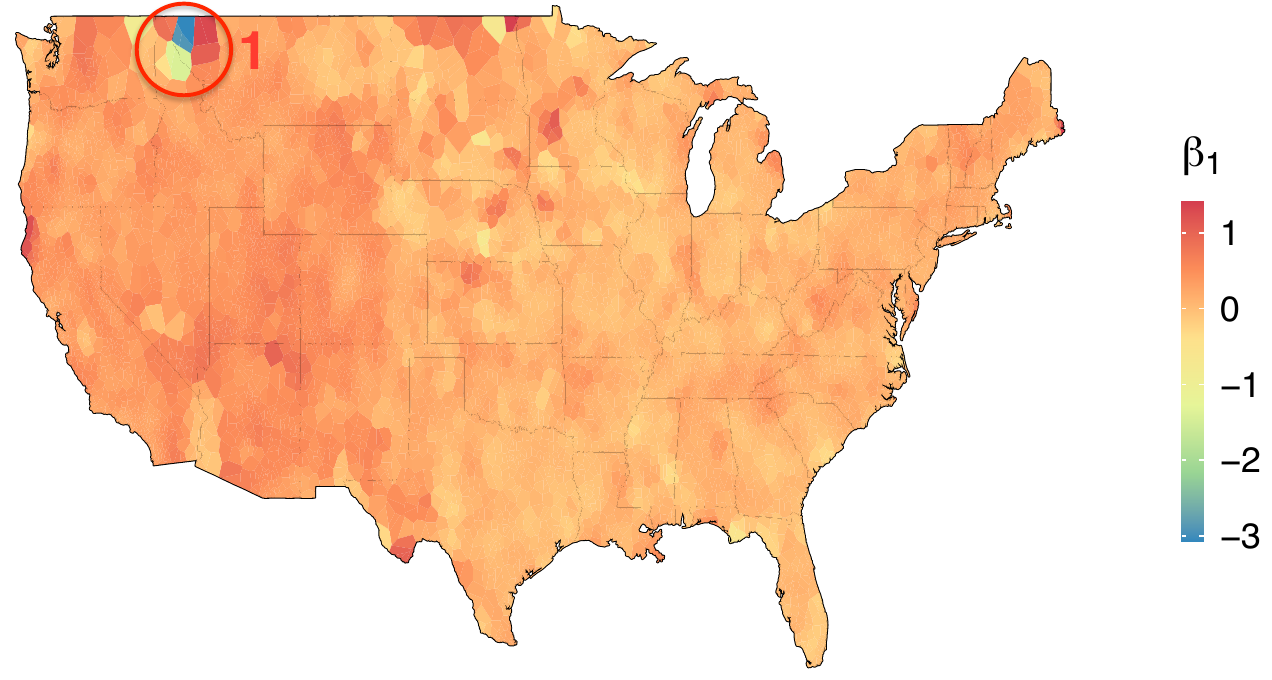}\\
        (b) Estimated coefficient of GWR
        \end{tabular}
    \end{minipage}\\
    \begin{minipage}{0.43\textwidth}
    \vspace{.2in}
    \centering
        \begin{tabular}{c}  
        \includegraphics[width = 1\textwidth]{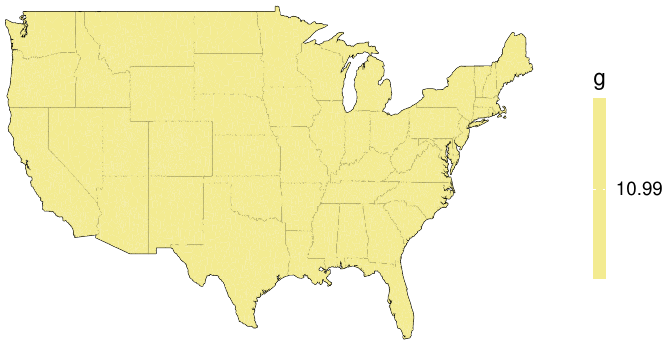}\\
        (c) Estimated intercept of PSCCM
        \end{tabular}
    \end{minipage}
    \hspace{0.3in}
    \begin{minipage}{0.43\textwidth}
    \vspace{.2in}
    \centering
        \begin{tabular}{c}  
        \includegraphics[width = 1\textwidth]{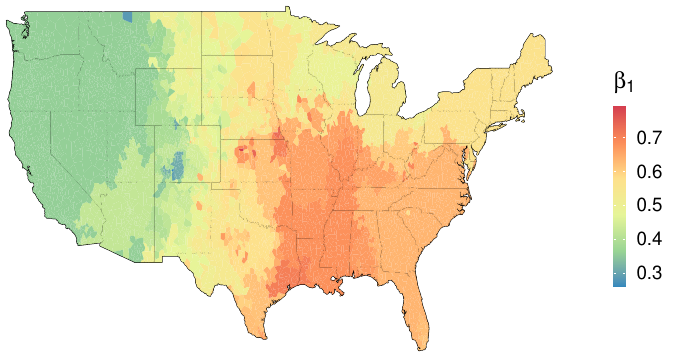}\\
        (d) Estimated coefficient of PSCCM
        \end{tabular}
    \end{minipage}
    \caption{The estimations of GWR and PSCCM on data of 2010-06-21.}\label{Fig: GWR_PSCCM_0621}
\end{figure}


\section{Discussion}\label{Section: conclusion}

In this work, we study the spatial coefficient clustering problem and develop a spatial heterogeneous additive partially linear model.
In our model, the spatially contiguous clusters are revealed by the coefficients of the linear components while a nonlinear component is added to deal with the spatially smoothing effect.
The method of bivariate spline over triangulation is adopted to approximate the nonlinear component, and a novel penalized method named the forest lasso is proposed for the spatial homogeneity pursuit of the linear coefficients.
With the proper linear transformation, we can efficiently estimate the model parameter and cluster the spatial locations simultaneously by solving a doubly penalized least square problem.

This study can be further extended in the following aspects.
First, our model only considers the intercept to be spatially smoothing, while the linear coefficients for all the covariates are spatially piece-wise constants. 
In many real applications, part of covariates might also have spatially smoothly varying effects over the domain. 
Thus, it would be interesting to propose a method that can also identify whether the effect of a covariate is spatially piece-wise constant or smoothly varying.
Secondly, as shown in our real data study, the spatial cluster pattern varies across days.
We can extend our model to study the spatial-temporal cluster structure on the regression coefficients. 
Furthermore, it is also important to develop statistical inference methods for the hypothesis testing on the clustered coefficients.


\bigskip
\begin{center}
{\large\bf SUPPLEMENTARY MATERIAL}
\end{center}

\section{Assumption}

In this section, we state the technical assumptions for our theoretical analysis.  
 \begin{itemize}
 \item[(A1).] The true non-parametric function $g \in  \mathcal{W}^{d+1,\infty}(\Omega)$, where $\mathcal{W}^{d+1,\infty}(\Omega)$ is the functional space $\{\alpha: |\alpha|_{k,\infty,\Omega} < \infty,
0 \leq k \leq d+1\}$ and $|\alpha|_{\upsilon,\infty,\Omega}=\max_{i+j=\upsilon}\Vert \nabla_{s_{1}}^{i}\nabla_{s_{2}}^{j}\alpha\Vert_{\infty ,\Omega}$.
 \item[(A2).] The density function $f_{\mathbf{s}}(\cdot)$ of $\mathbf{S}$ satisfies
\begin{align}
0 < c_f \leq \inf_{\mathbf{s} \in \Omega } f_{\mathbf{s}}(\mathbf{s}) \leq \sup_{\mathbf{s} \in \Omega } f_{\mathbf{s}}(\mathbf{s}) \leq C_f < \infty.
\end{align}
 The density function $f_{\mathbf{s}}(\cdot)$ of $\mathbf{S}$ is bounded away from zero and infinity on $\Omega$.
 \item[(A3).]  The random variables $X_{i\ell}$ are uniformly  bounded for $i=1,\ldots, n$, $\ell=1,\ldots,np$. Denote that $\lambda_1(\bs{s})\leq\cdots\leq\lambda_{p'+1}(\bs{s})$ be the the eigenvalues of $\mathrm{E} \left\{(1, \widetilde{\X}_{\mathcal{N}}^\top)^{\top}(1, \widetilde{\X}_{\mathcal{N}}^\top)|\mathbf{s}_i=\bs{s} \right\}$. 
 \item[(A4).] The triangulation $\triangle$ is $\pi$-quasi-uniform, i.e.,
$(\min_{ \tau \in \triangle}R_{\tau})^{-1} |\triangle| \leq \pi$ for some positive constant $\pi$.
\item[(A5).] The error vector $\bs{\varepsilon}=(\varepsilon_{1},\ldots,\varepsilon_{n})^{\top}$ has sub-Gaussian tails such that $\mathbb{P}(\bs{a}^\top\bs{\varepsilon} > \|\bs{a}\|x) \leq 2 \exp(-c_1x^2)$ for any vector $\bs{a}$ and $x >0$, where $0 < c_1 < \infty$.
\item[(A6).] The roughness tuning parameter $\rho$ satisfies that $\rho |\triangle|^{-4} \to 0$.
\item[(A7).] For $k \in \mathcal{N}^c \cap \mathcal{A}$, the adaptive weights $\omega_k$, tuning parameters $\lambda$ and $\rho$ satisfy $ |\triangle|^{(d+1)} + n^{-1/2} |\triangle|^{-1} + \rho |\triangle|^{-4} \ll \lambda \omega_k $ and $ |\triangle|^{-1}|\mathcal{N}\cap \mathcal{A}|^{1/2}  \ll \omega_k$. There exist some positive constant $c$, such that $\omega_k \geq c$, $k \in \mathcal{N}\cap \mathcal{A}$.
 \end{itemize}

Assumption (A1)--(A4) are the regular assumptions for analyzing the bivariate spline method over triangulation \citep{lai2013bivariate,yu2019estimation,wang2020simultaneous}. 
Assumption (A5) is the sub-Gaussian assumption for the measurement error, which are standard for high dimensional statistics analysis \citep{buhlmann2011statistics}. Assumption (A6) states the requirement for the roughness parameter.
Assumption (A7) are the order requirements on initial estimators for the selection consistency, which are also adopted in analyzing the adaptive Lasso method on high-dimensional non/semi-parametric models \citep{li2019additive,li2020sparse}.

\section{Notations}
We give the following notations, which will be used in our analysis.
Recall that $u_{kj}(x_k)$, $j \in \mathcal{J}$, are the original B-spline basis functions for the $k$th covariate, where $\mathcal{J}$ is the index set of the basis functions. In the following we define their centered basis $u_{kj}^0(x_k)$ and the standardized basis $U_{kj}(x_k)$.  Let $c_{kj}=\langle u_{kj},1\rangle$, we have
\begin{align}
u_{kj}^0(x_k)=u_{kj}(x_k)-\frac{c_{kj}}{c_{k1}}u_{k1}(x_k), \qquad
U_{kj}(x_k)=\frac{u_{kj}^{0}(x_k)}{\|u_{kj}^{0}\|},~j \in \mathcal{J},
\end{align}
so that $\mathbb{E}U_{kj}(X_k)=0$ and $\mathbb{E}U_{kj}^2(X_k)=1$. Similarly, we define the standardized Bernstein basis polynomials as $B_m^{\ast}(\mathbf{s})=B_m(\mathbf{s})/\|B_m\|$, $m \in \mathcal{M}$,
where $\mathcal{M}$ is the index set of Bernstein basis functions. For example, for bivariate spline space $\mathbb{S}_d^{r}(\triangle)$ containing $N$ triangles, $\mathcal{M}=\{1,2, \cdots, \frac{(d+1)(d+2)N}{2}\}$. Define the approximate space as
\begin{align}
  \mathcal{C}\!=\!\Big\{ \phi: \phi (\mathbf{x}, \mathbf{s})\!= \!\sum_{k=1}^{p}\sum_{j \in \mathcal{J}} \theta_{kj} U_{kj}(x_k)\!+\!\!\! \sum_{m \in \mathcal{M}}\gamma_mB^{\ast}_m(\mathbf{s}), x_k \in [0,1],  \mathbf{s} \in \Omega, \theta_{kj}, \gamma_m \in \mathbb{R} \Big\}.
\end{align}
 Given a triangle $\tau \in \triangle$, let $R_{\tau}$ be the radius of the largest disk contained in $\tau$, and let $|\tau|$ be length of the the longest edge. Define the shape parameter of $\tau$ as the ratio $\pi_{\tau}=|\tau|/R_{\tau}$. Note that when $\pi_{\tau}$ is small, the triangles are relatively uniform in the sense that all angles of triangles in the triangulation $\triangle$ are relatively the same. Denote the size of $\triangle$ by $|\triangle|:=\max \{|\tau|,\tau \in \triangle \}$. For any bivariate function $g: \Omega\to R$, denote $|g|_{\upsilon,\infty,\Omega}=\max_{i+j=\upsilon}\Vert \nabla_{s_{1}}^{i}\nabla_{s_{2}}^{j}g\Vert_{\infty ,\Omega}$.

\section{Supporting Lemmas}

\begin{lem}
\label{LEM:normequity}
Under Assumption (A4), there exist positive constants $c_s$, $C_s$, such that,
\begin{align}
	c_s\sum_{m\in \mathcal{M}}\psi_{m}^{2}\leq
	\Big\|\sum_{m\in \mathcal{M}} \psi_{m}B_{m}^{\ast}\Big\|_{L_2}^{2}\leq C_s\sum_{m\in \mathcal{M}}\psi_{m}^{2}.
\end{align}
\end{lem}

\begin{proof}
Notice that $\Big\|\sum_{m\in \mathcal{M}} \psi_{m}B_{m}^{\ast}\Big\|_{L_2}^{2} = \Big\|\sum_{m\in \mathcal{M}} (\psi_{m}\|B_{m}\|^{-1})B_{m}\Big\|_{L_2}^{2}$. According to Lemma 1 in \cite{lai2013bivariate}, we have 
\[
c_1|\triangle|^2\sum_{m\in \mathcal{M}}\|B_{m}\|^{-2}\psi_{m}^{2}\leq
	\Big\|\sum_{m\in \mathcal{M}} (\psi_{m}\|B_{m}\|^{-1})B_{m}\Big\|_{L_2}^{2}\leq C_2|\triangle|^2\sum_{m\in \mathcal{M}}\|B_{m}\|^{-2}\psi_{m}^{2}.
\]
Notice that $\|B_{m}\| \asymp |\triangle|$, $m \in \mathcal{M}$. Therefore, Lemma \ref{LEM:normequity} holds.
\end{proof}

\begin{lem}[Theorem 10.2, \cite{lai2007spline}]
\label{LEM:apporBi}
Suppose that $|\triangle|$ is a $\pi$-quasi-uniform triangulation of a polygonal domain $\Omega$, and $g(\cdot) \in  \mathcal{W}^{d+1,\infty}(\Omega)$.
\begin{itemize}
\item[(i)] For bi-integer $(a_{1},a_{2})$ with $0\leq {a_{1}}+{a_{2}} \leq d $, there exists a spline $g^{\ast}(\cdot)\in \mathbb{S}_{d}^{0}(\triangle)$ such that $\Vert \nabla_{z_{1}}^{a_{1}}\nabla_{z_{2}}^{a_{2}}\left(g-g^{\ast}\right) \Vert_{\infty}\leq C|\triangle|^{d+1-a_{1}-a_{2}}|\psi|_{d+1,\infty}$, where $C$ is a constant depending on $d$, and the shape parameter $\pi$.
\item[(ii)] For bi-integer $(a_{1},a_{2})$ with $0\leq {a_{1}}+{a_{2}} \leq d $, there exists a spline function $g^{\ast\ast}(\cdot)\in \mathbb{S}_{d}^{r}(\triangle)$ ($d\geq 3r+2$) such that $\Vert \nabla_{z_{1}}^{a_{1}}\nabla_{z_{2}}^{a_{2}}\left(g-g^{\ast\ast}\right) \Vert_{\infty}\leq C|\triangle|^{d+1-a_{1}-a_{2}}|g|_{d+1,\infty}$, where $C$ is a constant depending on $d$, $r$, and the shape parameter $\pi$.
\end{itemize}
\end{lem}
Lemma \ref{LEM:apporBi} shows that $\mathbb{S}_{d}^{0}(\triangle)$ has full approximation power, and $\mathbb{S}_{d}^{r}(\triangle)$ also has full approximation power if $d\geq 3r+2$.

\begin{lem}
\label{LEM:GammaInverse} Denote
$\bs{\Gamma}_{n, \rho}= n^{-1}(\mathbf{E}^{\top} \mathbf{E} + 2 n \rho \widetilde{\D}).$ Under Assumptions (A2) -- (A4), and (A6), there exist constants $0<c_{\Gamma}< C_{\Gamma}<\infty$, such that 
\begin{align}
c_{\Gamma} \leq \rho_{\min}(\bs{\Gamma}_{n, \rho}) \leq \rho_{\max}(\bs{\Gamma}_{n, \rho}) \leq C_{\Gamma}, 
\end{align}
with probability approaching to one, for large enough $n$.
\end{lem}

\begin{proof} 
First, we prove that there exist some constant such that $c_{\Gamma} \leq \rho_{\min}(\bs{\Gamma}_{n, 0}) \leq \rho_{\max}(\bs{\Gamma}_{n, 0}) \leq C_{\Gamma}$ holds with probability approaching to one. Notice that 
\begin{align*}
 (\bs{\theta}_{\mathcal{N}}^{\top}, \bs{\psi}^{\top})\bs{\Gamma}_{n, 0}(\bs{\theta}_{\mathcal{N}}^{\top}, \bs{\gamma}^{\top})^{\top} & = \frac{1}{n} \sum_{i=1}^n\big\{ \sum_{\ell = 1}^{p'} \theta_j X_{ij} + \sum_{m \in \mathcal{M}} \psi_{m}B_{m}^{\ast}(\mathbf{s}_i) \big\}^2 \\
 & = \mathrm{E}\big\{ \sum_{\ell = 1}^{p'} \theta_j X_{ij} + \sum_{m \in \mathcal{M}} \psi_{m}B_{m}^{\ast}(\mathbf{s}_i) \big\}^2 +  o_P(1).
\end{align*}
Denote 
$
\bs{a}(\mathbf{s})= \left\{ \sum_{m \in \mathcal{M}} \psi_{m}B_{m}^{\ast}(\mathbf{s}),  \bs{\theta}_{\mathcal{N}}^{\top} \right\}, ~
\bs{\Upsilon}(\bs{s}) = \mathrm{E}\left\{(1, \widetilde{\mathbf{X}}_{\mathcal{N}}^\top)^\top (1, \widetilde{\mathbf{X}}_{\mathcal{N}}^\top)\Big|\mathbf{s} = \bs{s}\right\}.
$ By Assumptions (A2), (A3) and (A4) and Lemma \ref{LEM:normequity}, we have
\begin{align*}
& \mathrm{E} \Big[ \mathrm{E} \big\{ \sum_{\ell = 1}^{p'} \theta_j X_j + \sum_{m \in \mathcal{M}} \psi_{m}B_{m}^{\ast}(\mathbf{s}) \big\}^2 \Big| \mathbf{s} \Big]= \mathrm{E} \left\{\bs{a}(\mathbf{s}) \bs{\Upsilon} (\mathbf{s} ) \bs{a}(\mathbf{s})^\top \right\}\\
&\leq C\sum_{j=1}^{p'} \theta_j^2  + C \left\Vert\sum_{m \in \mathcal{M}} \psi_{m}B_{m}^{\ast}\right\Vert^{2}_{L_2} \leq C_{\Gamma} \left(\|\bs{\theta}_{\mathcal{N}}\|^2 + \|\bs{\psi} \|^2\right).
\end{align*}
Similarly, we have $\mathrm{E}\big\{ \sum_{\ell = 1}^{p'} \theta_j X_{ij} + \sum_{m \in \mathcal{M}} \psi_{m}B_{m}^{\ast}(\mathbf{s}_i) \big\}^2 \geq  c_{\Gamma}  \left(\|\bs{\theta}_{\mathcal{N}}\|^2 + \|\bs{\psi} \|^2\right)$.  Then, $c_{\Gamma} \leq \rho_{\min}(\bs{\Gamma}_{n, 0}) \leq \rho_{\max}(\bs{\Gamma}_{n, 0}) \leq C_{\Gamma}$  follows. Similar to the arguments in Lemma B.8 in the supplementary materials in \cite{yu2019estimation}, we have $ (\bs{\theta}_{\mathcal{N}}^{\top}, \bs{\psi}^{\top})\widetilde{\D}(\bs{\theta}_{\mathcal{N}}^{\top}, \bs{\psi}^{\top})^{\top} = |\triangle|^{-4}\left(\|\bs{\theta}_{\mathcal{N}}\|^2 + \|\bs{\psi} \|^2\right)$. Therefore, under Assumption (A6), Lemma \ref{LEM:GammaInverse} holds.
\end{proof}


\section{Proof of Lemma \ref{LEM:oracle}}
\begin{proof}
According to Lemma \ref{LEM:apporBi}, there exists $\bs{\psi}^{\ast}$ satisfying $\sup_{\bs{s} \in \Omega} |g(\bs{s})-\sum_{m \in \mathcal{M}}B_m^\ast(\bs{s})\psi_m^{\ast}| = |\triangle|^{(d+1)}$. By the definition of the oracle estimator, it is easy to verify that 
\begin{align}
(\widehat{\bs{\psi}}^o, \widehat{\bs{\theta}}^o) - (\bs{\psi}^{\ast}, \bs{\theta}_{\mathcal{N}}) =n^{-1} \bs{\Gamma}_{n, \rho}^{-1}\left[\mathbf{E}^\top \{\mathbf{Y} - (\bs{\psi}^{\ast \top}, \bs{\theta}_{\mathcal{N}}^{\top})\mathbf{E}\} + 2 n \rho \widetilde{\D} (\bs{\psi}^{\ast \top}, \bs{\theta}_{\mathcal{N}}^{\top})^{\top}\right],
\end{align}
and thus,
\begin{align}
\|(\widehat{\bs{\psi}}^o, \widehat{\bs{\theta}}^o) - (\bs{\psi}^{\ast}, \bs{\theta}_{\mathcal{N}})\| \leq C n^{-1} \left\{ \|\mathbf{E}^\top \bs{\delta}\|+ \|\mathbf{E}^\top \bs{\epsilon}\| + 2 n \rho \|\widetilde{\D} (\bs{\psi}^{\ast \top}, \bs{\theta}_{\mathcal{N}}^{\top})^{\top}\| \right\},
\end{align}
where $\bs{\delta}^{\top} = \{g(\mathbf{s}_i) - \sum_{m \in \mathcal{M}}B_m^\ast(\mathbf{s}_i)\psi_m^{\ast}\}_{i=1}^n$. 
Notice that $\|\mathbf{E}^\top \bs{\delta}\| = O_P\{n|\triangle|^{(d+1)}\}$ and $\textrm{E}\left(\|\mathbf{E}^\top \bs{\epsilon}\|^2\right) = \textrm{tr}(\mathbf{E} \mathbf{E}^\top)$ and $\|\mathbf{E}^\top \bs{\epsilon}\| =O_P(n^{1/2} |\triangle|^{-1})$. 
According to the results in Lemma B.9 in \cite{yu2019estimation}, $n \rho \|\widetilde{\D} (\bs{\psi}^{\ast \top}, \bs{\theta}_{\mathcal{N}}^{\top})^{\top}\| = O(n \rho  |\triangle|^{-4})$. Therefore, we have 
$
\|(\widehat{\bs{\psi}}^o, \widehat{\bs{\theta}}^o) - (\bs{\psi}^{\ast}, \bs{\theta}_{\mathcal{N}})\| = O_P(|\triangle|^{(d+1)} + n^{-1/2} |\triangle|^{-1} + \rho |\triangle|^{-4}).
$
Combining the conclusion in Lemma \ref{LEM:normequity}, we establish Theorem \ref{THE:consistency}.
\end{proof}

\section{Proof of Theorem \ref{THE:consistency}}
\begin{proof}
Given a specific spanning tree, we have the following minimization problem after the linear transformation:
\begin{align}
(\hat{\vpsi},\hat{\vtheta}) = \arg\min_{\vpsi,\vtheta} 
\frac{1}{2n}\|\y - \Bt\vpsi - \widetilde{\X} \vtheta\|_2^2 + \rho \vpsi^\top \D \vpsi + \lambda \sum_{l \in \Ac} |[\vtheta]_l|.
\label{EQU:psitheta}
\end{align}
Notice that any $(\vpsi, \vtheta)$ satisfies the following KKT conditions: 
\begin{itemize}
\item[C1.] $\Bt^{\top}(\y - \Bt\vpsi - \widetilde{\X} \vtheta) {{-}} 2 n \rho  \D \vpsi  = \mathbf{0},$
\item[C2.] For $k \notin \mathcal{A}$, $\widetilde{\X}_k^{\top}(\y - \Bt\vpsi - \widetilde{\X} \vtheta)   = 0,$
\item[C3.] For $k \in \mathcal{A}$, $\widetilde{\X}_k^{\top}(\y - \Bt\vpsi - \widetilde{\X} \vtheta)   =  n \lambda \tau_k$, where $\tau_k = \textrm{sign} (\theta_k)$, if $\theta_k \neq 0$, and $|\tau_k| \leq 1$, if $\theta_k = 0$.
\end{itemize}
is the unique minimizer of (\ref{EQU:psitheta}).

Denote vector $\bs{v} = (v_m, m = 1, \ldots, p)^{\top}$ with elements $v_m = \omega_k\textrm{sign} (\widehat{\theta}^o_m) I(m \in \mathcal{A} \cap \mathcal{N})$. Then we define $\widehat{\bs{\eta}} = [\{\bs{\Gamma}_{n, \rho}^{-1} (\mathbf{E}^\top \y - \lambda n \bs{v})\}^{\top}, \mathbf{0}^{\top}]^\top$. In the following, we prove that $\widehat{\bs{\eta}}$ satisfies KKT conditions with probability approaching one. Plugging in $\widehat{\bs{\eta}}$ to $\mathbf{E}^{\top}(\y - \Bt\vpsi - \widetilde{\X} \vtheta) - 2 n \rho \widetilde{\D} \vpsi $, we have 
\begin{align}
\mathbf{E}^{\top}\{\y - \mathbf{E}\bs{\Gamma}_{n, \rho}^{-1}&(\mathbf{E}^\top \y - \lambda n \bs{v})\}- 2 n \rho  \widetilde{\D}   \bs{\Gamma}_{n, \rho}^{-1}(\mathbf{E}^\top \y - \lambda n \bs{v}) 
 = \lambda n \bs{v}.
\end{align}
According to the definition of $\bs{v}$, it implies that $\widehat{\bs{\eta}}$ satisfies Conditions C1 and C2 and C3 for $k \in \mathcal{N} \cap \mathcal{A}$. 

Next, we prove for $k \in \mathcal{N}^c \cap \mathcal{A}$, $|\widetilde{\X}_k^{\top}(\y - \Bt\vpsi - \widetilde{\X} \vtheta)| \leq \lambda \omega_k  n$. According Assumption (A5), 
\begin{align}\label{Eq: proof_them1_1}
&
\mathbb{P}  \{|\widetilde{\X}_k^{\top}(\y - \Bt\vpsi - \widetilde{\X} \vtheta)| > \lambda \omega_k  n, \exists k \in \mathcal{N}^c \cap \mathcal{A}\}\notag \\
\le 
&
\mathbb{P}  \{|\widetilde{\X}_k^{\top} \bs{\epsilon}| > \lambda \omega_k  n/3, \exists k \in \mathcal{N}^c \cap \mathcal{A}\} +  \mathbb{P}  \{|\widetilde{\X}_k^{\top} \bs{\delta}| > \lambda \omega_k  n/3, \exists k \in \mathcal{N}^c \cap \mathcal{A}\} \notag \\
&  + \mathbb{P} \left\{ |\widetilde{\X}_k^{\top}\bs{\Gamma}_{n, \rho}^{-1} \bs{v}| > \omega_k/3, \exists k \in \mathcal{N}^c \cap \mathcal{A}\right\}\notag \\
\le
&
\sum_{k \in \mathcal{N}^c \cap \mathcal{A}} \mathbb{P}  \{|\widetilde{\X}_k^{\top} \bs{\epsilon}| > \lambda n \omega_k /3\} + \mathbb{P}  \{|\widetilde{\X}_k^{\top} \bs{\delta}| > \lambda n \omega_k/3\} + \mathbb{P} \{ |\widetilde{\X}_k^{\top}\mathbf{E}\bs{\Gamma}_{n, \rho}^{-1} \bs{v}| > \omega_k/3\}
\end{align}
where $ \bs{\delta} = \bs{\mu} - n ^{-1}\mathbf{E} \bs{\Gamma}_{n, \rho}^{-1}  \mathbf{E}^{\top}(\bs{\mu} + \bs{\epsilon})$ and $\bs{\mu} = \{\mu_1, \ldots, \mu_n\}^{\top}$ being the expectation of $\y$.

For each term in (\ref{Eq: proof_them1_1}), we give the following bounds. 
Note that $\|\widetilde{\X}_k\| \le {n^{1/2}}C$ for $k \in \mathcal{N}^c \cap \mathcal{A}$ and some positive constant $C$. 
By the sub-Gaussian assumption, we have $\mathbb{P}  \{|\widetilde{\X}_k^{\top} \bs{\epsilon}| > \lambda \omega_k  n/3\} \le {2 \exp(-c_1n\lambda^2\omega_k^2)} $. 
The results in Lemma \ref{LEM:oracle} implies  $n^{-1/2}\|\bs{\delta}\| = O_P\{|\triangle|^{(d+1)} + n^{-1/2} |\triangle|^{-1} + \rho |\triangle|^{-4}\}$. 
By Assumption (A7) that $\lambda\omega_k \gg |\triangle|^{(d+1)} + n^{-1/2} |\triangle|^{-1} + \rho |\triangle|^{-4}$, we have $\mathbb{P}  \{|\widetilde{\X}_k^{\top} \bs{\delta}| > \lambda \omega_k  n/3,  \exists k \in \mathcal{N}^c \cap \mathcal{A}\} =o_P(1)$. 
According the results in Lemma \ref{LEM:GammaInverse}, the eigenvalues of $n\bs{\Gamma}_{n, \rho}^{-1}$ are bounded by some positive constants $C_1$ and $C_2$. Therefore, the elements in the vector $\mathbf{E} \bs{\Gamma}_{n, \rho}^{-1}\bs{v}$ are with the order of $O_P(n^{-1} |\triangle|^{-1}|\mathcal{N}\cap \mathcal{A}|^{1/2})$, which implies $\widetilde{\X}_k^\top\mathbf{E}\bs{\Gamma}_{n, \rho}^{-1} \bs{v} = O_P(|\triangle|^{-1}|\mathcal{N}\cap \mathcal{A}|^{1/2})$, $\forall k \in \mathcal{N}^c \cap \mathcal{A}$.
Also with Assumption (A7) that $ |\triangle|^{-1}|\mathcal{N}\cap \mathcal{A}|^{1/2}  \ll \omega_k$, it holds $\mathbb{P} \{ |\widetilde{\X}_k^{\top}\mathbf{E}\bs{\Gamma}_{n, \rho}^{-1} \bs{v}| > \omega_k/3\} = o_p(1)$.
To summarize, we have
\begin{align}
&
\mathbb{P}  \{|\widetilde{\X}_k^{\top}(\y - \Bt\vpsi - \widetilde{\X} \vtheta)| > \lambda \omega_k  n, \exists k \in \mathcal{N}^c \cap \mathcal{A}\} \notag\\
\le 
&
\sum_{k \in \mathcal{N}^c \cap \mathcal{A}} 2 \exp(-c_1n\lambda^2\omega_k^2) + o_p(1) = o_P(1),
\end{align}
Hence, $\widehat{\bs{\eta}}$ satisfies Conditions C3 for $k \in \mathcal{N}^c \cap \mathcal{A}$. 
And we prove that $\widehat{\bs{\eta}}$ is the minimizer of objective function (\ref{EQU:psitheta}).

Furthermore, note that $\|\widehat{\bs{\eta}} - (\widehat{\bs{\psi}}^{o \top}, \widehat{\bs{\theta}}^{o \top}, \mathbf{0}^{\top})^{\top} \| = \| \lambda n\bs{\Gamma}_{n, \rho}^{-1} \bs{v} \| = O_P(\lambda |\mathcal{N}\cap \mathcal{A}|^{1/2})$. Combining the results in Lemma \ref{LEM:oracle}, Theorem \ref{THE:consistency} holds.
\end{proof}






\bibliographystyle{apalike}
\bibliography{reference}
\end{document}